\documentclass[
  aps,
  pra,
  twocolumn,
  notitlepage,
  superscriptaddress,
  longbibliography
]{revtex4-2}

\usepackage[T1]{fontenc}
\usepackage{amsmath,amssymb,amsthm,amsfonts}
\usepackage{bm}
\usepackage{mathtools}
\usepackage{graphicx}
\usepackage{booktabs}
\usepackage{xcolor}
\usepackage{hyperref}
\hypersetup{hidelinks}

\newcommand{\dd}{\mathrm{d}}
\newcommand{\Tr}{\operatorname{Tr}}
\newcommand{\E}{\mathbb{E}}
\newcommand{\nth}{n_{\mathrm{th}}}

\newtheorem{theorem}{Theorem}[section]

\theoremstyle{remark}

\begin{document}

\title{Exact Conditional Momentum Moments for Nonlinear Homodyne Trajectories}

\author{Jacob Emerson}
\affiliation{Department of Electrical and Computer Engineering, Princeton University}

\date{\today}

\begin{abstract}
Simulation of nonlinear stochastic master equations generally requires
trajectory-level evolution of large Hilbert-space density matrices, making
strongly nonlinear continuously monitored systems computationally challenging.
Here we identify an exactly solvable class centered on polynomial Hamiltonians
of commuting quadratures with linear damping and homodyne measurement
of those same quadratures. For Gaussian initial states, or more generally for initial states with Gaussian measured-quadrature marginal and
polynomial conditional-momentum sections, we derive an exact finite
stochastic representation through any fixed momentum order. The result extends to
interacting multimode systems, arbitrary measurement efficiency, and finite
thermal occupation. The representation requires no Fock-space truncation,
Gaussian approximation, or approximate moment closure. For Gaussian inputs,
its dimension grows polynomially with mode number at fixed Hamiltonian
degree and momentum order, and the cost is linear in the number of required time steps. As an application, we derive an exact causal filter for monitored cubic-phase
gates that uses the homodyne record to determine the variance-minimizing momentum
displacement. Numerical simulations validate the representation against direct
stochastic-master-equation evolution, demonstrate substantial speedups over
converged Fock-space propagation, and show that this record-conditioned recentering
can efficiently recover nonlinear squeezing otherwise lost through trajectory averaging.
\end{abstract}
\maketitle
\section{Introduction}
\label{sec:introduction}

Continuous quantum measurement describes quantum systems whose evolution is
conditioned on a continuously acquired measurement record
\cite{belavkin1992,wiseman_milburn2009,jacobs2014}. It underlies quantum
feedback and state estimation \cite{wiseman_milburn2009,jacobs2014}, quantum
sensing \cite{degen2017}, and measurement-based quantum information processing
\cite{wiseman_milburn2009,weedbrook2012gaussian}. Each realization of the
record generates a distinct conditional trajectory governed by a stochastic
master equation (SME), in which deterministic open-system dynamics are
accompanied by stochastic measurement backaction.

Exact reductions of quantum dynamics are known when Gaussian or Lie-algebraic
structure can be exploited. Linear bosonic systems admit finite descriptions
in terms of means and covariances
\cite{weedbrook2012gaussian,serafini2017quantum}, while Lie-algebraic and
phase-space methods yield exact results for certain nonlinear open systems
\cite{Chaturvedi1991,Emerson2026Kerr}. Structured state representations,
including Gaussian and coherent-state decompositions, have also enabled
efficient simulation of certain non-Gaussian bosonic systems
\cite{Bourassa2021,Yao2024,MarshallAnand2023,
DiasKonig2024,Chabaud2021}. For continuously monitored nonlinear
systems, projection filters obtain tractable approximations by restricting the
conditional state to a finite-dimensional manifold
\cite{vanHandelMabuchi2005}. Exact phase-space formulations of nonlinear
quantum filtering have also been derived, but generally remain stochastic
integro-differential equations \cite{Vladimirov2017}. More recently, exact
model reduction has been developed for selected expectation values of
finite-dimensional quantum filters \cite{Grigoletto2025}, while exact
conditional propagators are known for linear Heisenberg-picture dynamics
\cite{Warszawski2020}.

These examples show that complicated quantum dynamics can sometimes be
greatly simplified when expressed in the right variables.
Whether similarly broad simplifications exist for generic continuously monitored
nonlinear systems remains an open problem. A natural approach is to
describe the evolution using a finite set of observable moments.
Under continuous measurement, however, the stochastic innovation
couples these moments to progressively higher-order conditional
correlations. Even when the deterministic moment equations close,
the corresponding stochastic equations need not. This raises the
question of whether the obstruction to closure is inherent to the
dynamics or instead a consequence of the variables used to describe them.

In this work, we identify a class of Hamiltonians for which this
obstruction can be overcome exactly. Polynomial Hamiltonians of
commuting quadratures, \(F(\bm Q,t)\), modify the momentum statistics
without changing the measured-quadrature marginal. Such Hamiltonians
describe nonlinear phase operations and arise as effective interactions
in engineered bosonic systems. Their structure is closely related to
the low-symplectic-coherence regime, where even strongly nonlinear
continuous-variable dynamics can remain classically tractable
\cite{UpretiChabaud2025}.

To exploit this separation under continuous measurement, we introduce
the conditional momentum moments
\begin{equation}
m_n(q,t)
=
\frac{\displaystyle\int p^n W_c(q,p,t)\,\dd p}
     {\displaystyle\int W_c(q,p,t)\,\dd p},
\label{eq:introConditionalMomentum}
\end{equation}
where \(W_c(q,p,t)\) is the record-conditioned Wigner function.
Unlike ordinary observable moments, these quantities retain the
measured quadrature as an explicit coordinate while reducing the
momentum dependence to a hierarchy of moments. Normalization by the
measured-quadrature marginal cancels the multiplicative homodyne
innovation, leaving only differentiation with respect to \(q\).
Thus, whereas measurement introduces higher-order correlations into
the ordinary moment hierarchy, its action on polynomial conditional
momentum moments lowers their degree.

For Gaussian initial states, polynomial Hamiltonians of the measured
quadratures, linear thermal damping, and homodyne detection of those
same quadratures, we prove that the conditional momentum hierarchy
closes exactly. The result extends to interacting multimode systems,
compatible momentum-linear generators, and a broader class of
non-Gaussian initial states, including polynomial-phase states.
Moments of the form \(\langle Q^rP^n\rangle_c\) can therefore be
reconstructed through Gaussian integration from a finite set of coupled
It\^o equations, with the dimension of the hierarchy depending on the
momentum order \(n\) but not the position order \(r\).
We derive the asymptotic scaling of the hierarchy, establishing
polynomial rather than exponential growth in the number of modes
at fixed Hamiltonian degree and momentum order, and demonstrate
substantially more favorable scaling for structured multimode
interactions. The representation requires no Fock-space truncation,
and its simulation cost is linear in the number of time steps for
a fixed hierarchy. Shared-record comparisons with direct SME
simulations verify the reduced dynamics, while large-scale multimode
simulations demonstrate its computational advantages.

We apply the finite hierarchy to real-time filtering of polynomial-phase
resources. Our principal example is the cubic phase gate, a fundamental
non-Gaussian element of universal continuous-variable processing
\cite{Lloyd1999} that has motivated measurement-induced optical schemes and
nonlinear feedforward protocols \cite{Marek2011,Konno2021,Sakaguchi2023}, as
well as a recent native implementation in a superconducting bosonic circuit
\cite{Eriksson2024}. For a polynomial phase operation
\(U_{\Phi}=\exp[-i\Phi(Q)]\), the associated nonlinear nullifier is
\begin{equation}
N_{\Phi}=P+\Phi'(Q).
\label{eq:introNullifier}
\end{equation}
Its variance, compared with the optimized Gaussian bound, quantifies
nonlinear squeezing \cite{Marek2011,Konno2021,Kala2022}. Different homodyne
records produce different conditional nullifier means, so discarding the
record adds their between-record variance. The closed hierarchy evaluates
the conditional mean causally from the homodyne current and efficiently supplies
the optimal momentum displacement that recenters each trajectory.

In a platform-independent, loss-dominated cubic benchmark, optimal recentering
reduces the normalized nonlinear-squeezing ratio from \(1.128\) to \(0.685\),
removing \(39.2\%\) of the variance and recovering a value below the optimized
Gaussian threshold. Shared-record benchmarks reproduce the mixed-state SME
endpoint moments with root-mean-square discrepancies below \(8.7\times10^{-4}\)
at the finest hierarchy time step, while reducing the per-record runtime by a
factor of $172\times$ for the tested implementation. A loss-limited projection
based on reported superconducting-circuit parameters predicts up to \(9.9\%\)
variance recovery and a \(26\%\) extension of the nonlinear-squeezing lifetime.
As dephasing is the dominant source of decoherence in the referenced
experiment, realizing these improvements would require substantially improved
phase coherence or independent correction of phase noise. The compatibility
of the present method with additional correction protocols is left for future work. Together, these results establish conditional
momentum moments as an exact, cutoff-independent route to real-time estimation
and record-conditioned control of selected nonlinear observables, replacing
full conditional-state propagation by a finite stochastic filter whenever only
the corresponding observable sector is required.

\section{Physical and mathematical framework}
\label{sec:framework}

\subsection{Continuously monitored phase dynamics}
\label{subsec:frameworkSME}

We first fix conventions for one bosonic mode.  Let
\begin{equation}
Q=\frac{a+a^\dagger}{\sqrt{2}},
\qquad
P=\frac{a-a^\dagger}{i\sqrt{2}},
\qquad
[Q,P]=i,
\label{eq:quadratureDefinitions}
\end{equation}
so the vacuum variances are
\(\operatorname{Var}(Q)=\operatorname{Var}(P)=1/2\).  Our starting point is
the diffusive stochastic master equation, written in It\^o form
\cite{wiseman_milburn2009,jacobs2014},
\begin{align}
\dd\rho_c={}&-i[H(t),\rho_c]\,\dd t
+\kappa(\nth+1)\mathcal D[a]\rho_c\,\dd t
\nonumber\\
&+\kappa\nth\mathcal D[a^\dagger]\rho_c\,\dd t
+\sqrt{\eta\kappa}\,\mathcal H[a]\rho_c\,\dd W_t,
\label{eq:sme}
\end{align}
where
\begin{align}
\mathcal D[L]\rho
&=L\rho L^\dagger-\frac12\{L^\dagger L,\rho\},
\label{eq:dissipator}\\
\mathcal H[L]\rho
&=L\rho+\rho L^\dagger
-\Tr[(L+L^\dagger)\rho]\rho.
\label{eq:innovationSuperoperator}
\end{align}
Here \(\kappa\) is the energy-decay rate, \(\nth\) is the thermal occupation,
and \(\eta\in[0,1]\) is the fraction of the vacuum-loss output captured by
the homodyne detector; the excess thermal noise is unmonitored.  We use
\(\langle A\rangle_c=\Tr(A\rho_c)\) for conditional expectations.  The
homodyne record then satisfies
\begin{equation}
\dd Y_t=\sqrt{2\eta\kappa}\,\langle Q\rangle_c\,\dd t+\dd W_t,
\label{eq:homodyneRecord}
\end{equation}
where \(\dd W_t\) is a Gaussian increment with mean zero and variance
\(\dd t\); in It\^o notation, \((\dd W_t)^2=\dd t\).  All expectations below
are conditioned on the record up to time \(t\).  Equivalently,
\begin{equation}
\dd W_t
=\dd Y_t-\sqrt{2\eta\kappa}\,\langle Q\rangle_c\,\dd t
\label{eq:feedbackInnovationFromRecord}
\end{equation}
is the innovation: the measured increment minus the signal predicted by the
current conditional state.  This same increment drives the stochastic term in
Eq.~\eqref{eq:sme}, so the SME updates the state as each new piece of the
record arrives.  The efficiency \(\eta\) controls the strength of this
record-dependent update.  Noise entering through unmonitored loss or thermal
channels still affects the state through the dissipative terms in
Eq.~\eqref{eq:sme}, but is not contained in \(\dd Y_t\) and cannot be inferred
from it alone.

The difficulty of evolving selected observables directly can already be seen
from the innovation term. For an arbitrary operator \(A\), the SME gives
\begin{align}
\dd\langle A\rangle_c
={}&
\langle\mathcal L_t^\dagger A\rangle_c\,\dd t
\nonumber\\
&+
\sqrt{\frac{\eta\kappa}{2}}
\left(
\langle\{A,Q\}\rangle_c
+i\langle[A,P]\rangle_c
-2\langle A\rangle_c\langle Q\rangle_c
\right)\dd W_t,
\label{eq:ordinaryMomentInnovation}
\end{align}
where \(\mathcal L_t^\dagger\) contains the Hamiltonian and dissipative
contributions. In particular, the stochastic update of the first moment
\(\langle P\rangle_c\) depends on the mixed second moment,
\begin{equation}
\left.\dd\langle P\rangle_c\right|_{\dd W}
=
\sqrt{\frac{\eta\kappa}{2}}
\left(
\langle\{Q,P\}\rangle_c
-2\langle Q\rangle_c\langle P\rangle_c
\right)\dd W_t.
\label{eq:ordinaryMomentumCoupling}
\end{equation}
The update of \(\langle\{Q,P\}\rangle_c\) then contains a symmetrically
ordered moment proportional to \(Q^2P\), whose update contains one
proportional to \(Q^3P\), and so on. Homodyne conditioning therefore produces
an upward hierarchy in the power of the measured quadrature. For linear
Gaussian dynamics these higher correlations remain fixed by the conditional
means and covariances, but nonquadratic phase Hamiltonians destroy that
Gaussian reduction. Direct moment propagation then generally requires a
closure approximation or full-state numerical simulation. The conditional
momentum construction developed below removes precisely this upward
stochastic coupling.

The single-mode Hamiltonian class we consider is
\begin{align}
H(t)
={}&F(Q,t)+g_0(t)P+\frac{g_1(t)}{2}\{Q,P\},
\label{eq:hamiltonianClass}\\
F(Q,t)
={}&\sum_{r=0}^{d}f_r(t)Q^r,
\qquad f_r(t)\in\mathbb R.
\label{eq:phasePolynomial}
\end{align}
The last two terms generate translations and dilations of the measured
quadrature and preserve the affine structure used below.  Higher powers of
\(P\), or \(Q\)-dependent coefficients multiplying \(P\) beyond the displayed
dilation term, generally do not.  The multimode extension is stated after the
single-mode construction.

\subsection{Interaction-picture and homodyne interpretation}
\label{subsec:interactionPicture}

The absence of the free oscillator term \(\omega a^\dagger a\) in
Eq.~\eqref{eq:hamiltonianClass} means that the Hamiltonian is naturally stated in the rotating frame; $Q$ and $P$ are rotating quadratures with frequency $\omega$.  Let \(R(t)=\exp(-i\omega t a^\dagger a)\) and define
\(\rho_R=R^\dagger\rho_{\mathrm{lab}}R\).  If
\(H_{\mathrm{lab}}=\omega a^\dagger a+V_{\mathrm{lab}}(t)\), direct
differentiation gives
\begin{equation}
H_R(t)=R^\dagger(t)V_{\mathrm{lab}}(t)R(t).
\label{eq:rotatingHamiltonian}
\end{equation}
The theorem therefore applies whenever this transformed Hamiltonian has the
form in Eq.~\eqref{eq:hamiltonianClass}.  Conversely, a target interaction
\(F(Q,t)\) corresponds formally to
\begin{align}
V_{\mathrm{lab}}(t)
&=R(t)F(Q,t)R^\dagger(t)
\nonumber\\
&=F(Q\cos\omega t-P\sin\omega t,t).
\label{eq:labFramePhaseInteraction}
\end{align}
Experimentally, resonant drives can be chosen so that rapidly oscillating
terms average away in Eq.~\eqref{eq:rotatingHamiltonian}, leaving the desired
effective interaction.  This is the construction used in the demonstrated
superconducting cubic phase gate of Ref.~\cite{Eriksson2024}, whereas separate measurement-induced ancilla-based methods exist for optical setups \cite{Miyata2016}.

For the case we consider here, measurement must be expressed in the same frame as the Hamiltonian.  Homodyne detection with
relative phase \(\theta\) measures
\(Q_\theta=Q\cos\theta+P\sin\theta\).  Locking the local-oscillator phase to
the oscillator reference keeps \(\theta\) constant in the rotating frame, and
we choose that fixed axis to be \(Q\) in Eq.~\eqref{eq:sme}.  A known constant
offset only relabels the two quadratures; an uncontrolled time-dependent phase
mixes them and lies outside the exact closure
\cite{wiseman_milburn2009,weedbrook2012gaussian}.

A static nonlinear potential in the lab frame generally transforms into a
function of both rotating-frame quadratures and therefore falls outside our class.  The result instead describes engineered interactions whose
rotating-frame Hamiltonian has Eq.~\eqref{eq:hamiltonianClass}, or time windows
over which that effective description is accurate.

\subsection{Polynomial phase nullifiers and record-based recentering}
\label{subsec:nullifierFramework}

Let the input be momentum squeezed, with
\begin{equation}
\operatorname{Var}_{\mathrm{in}}(P)=V_P<\frac{1}{2},
\end{equation}
where \(1/2\) is the vacuum variance in our quadrature convention.  This
reduced noise is the resource that a polynomial phase operation transfers to
a nonlinear quadrature.  For a \(Q\)-dependent Hamiltonian acting over
\(0\leq t\leq T\), define
\begin{equation}
U_{\Phi}=e^{-i\Phi(Q)},
\qquad
\Phi(Q)=\int_0^T F(Q,t)\,\dd t .
\end{equation}
Since
\begin{equation}
\bigl[P+\Phi'(Q)\bigr]U_{\Phi}=U_{\Phi}P,
\end{equation}
the output state
\(\lvert\psi_{\Phi}\rangle=U_{\Phi}\lvert\psi_{\mathrm{in}}\rangle\)
satisfies
\begin{equation}
\operatorname{Var}_{\psi_{\Phi}}\!\left[N_{\Phi}\right]
=V_P,
\qquad
N_{\Phi}=P+\Phi'(Q).
\label{eq:nullifierDefinition}
\end{equation}
Thus the input momentum squeezing appears after the nonlinear operation as
reduced fluctuations of \(N_{\Phi}\), called the polynomial-phase
\emph{nullifier}.  Its variance measures how well the nonlinear resource is
preserved and, when it falls below the optimized Gaussian bound discussed in Appendix \ref{app:gaussianNullifierBound}, witnesses
nonlinear squeezing \cite{Marek2018,Konno2021}.  For
\(\Phi(Q)=zQ^3/3\), the cubic-phase nullifier is
\(N_{\Phi}=P+zQ^2\).

During monitored open-system evolution, the nullifier statistics become
conditioned on the homodyne record.  In particular, the record determines the
conditional offset
\begin{equation}
\mu_t=\langle N_{\Phi}\rangle_c .
\label{eq:conditionalNullifierMean}
\end{equation}
A momentum displacement can remove this accessible offset, while leaving
attenuation and unobserved noise unchanged. The conditional-momentum
hierarchy developed below supplies \(\mu_t\) and the conditional nullifier
variance directly from the complete preceding record.  After establishing
the closure, we derive the optimal record-conditioned displacement and its
extension to time-dependent and multimode nullifiers.

\section{Exact closure of conditional momentum moments}
\label{sec:closure}

We now derive the finite conditional-momentum representation.  Integrating
over momentum isolates an autonomous Gaussian marginal for the measured
quadrature.  Normalizing the remaining momentum sections by this marginal
then converts the homodyne innovation into differentiation with respect to
\(q\), which preserves a finite polynomial space.

\subsection{Wigner dynamics and the measured marginal}
\label{subsec:wignerMarginal}

For the Hamiltonian in Eq.~\eqref{eq:hamiltonianClass}, the conditional Wigner
function obeys
\begin{align}
\dd W_c
={}&
\Bigg\{
-\partial_q\!\left[
\left(g_0+\left(g_1-\frac{\kappa}{2}\right)q\right)W_c
\right]
+\left(g_1+\frac{\kappa}{2}\right)\partial_p(pW_c)
\nonumber\\
&\quad
+\frac{\kappa}{4}(2\nth+1)
(\partial_q^2+\partial_p^2)W_c
\nonumber\\
&\quad
+\sum_{\ell=0}^{\lfloor(d-1)/2\rfloor}
\frac{(-1)^\ell}
{2^{2\ell}(2\ell+1)!}
F^{(2\ell+1)}(q,t)
\partial_p^{2\ell+1}W_c
\Bigg\}\dd t
\nonumber\\
&+
\sqrt{\frac{\eta\kappa}{2}}
\left[
2(q-\bar q_t)W_c+\partial_qW_c
\right]\dd W_t,
\label{eq:wignerSMEClosure}
\end{align}
where \(\bar q_t=\langle Q\rangle_c\) and
\(F^{(r)}=\partial_q^rF\).  The term with \(\ell=0\) is the classical force,
while \(\ell\geq1\) gives the finite sequence of Moyal corrections
\cite{moyal1949,hillery1984}.  We assume sufficient decay for the integrations
by parts below.  Appendix~\ref{app:innovationCancellation} derives the
stochastic term and the normalized It\^o quotient explicitly.

Define the measured-quadrature marginal
\begin{equation}
\rho(q,t)=\int_{-\infty}^{\infty}W_c(q,p,t)\,\dd p.
\label{eq:measuredMarginal}
\end{equation}
Every contribution from \(F\) is a total \(p\)-derivative and vanishes under
this integral, leaving
\begin{align}
\dd\rho
={}&
\Bigg\{
-\partial_q\!\left[
\left(g_0+\left(g_1-\frac{\kappa}{2}\right)q\right)\rho
\right]
+\frac{\kappa}{4}(2\nth+1)\partial_q^2\rho
\Bigg\}\dd t
\nonumber\\
&+
\sqrt{\frac{\eta\kappa}{2}}
\left[
2(q-\bar q_t)\rho+\partial_q\rho
\right]\dd W_t.
\label{eq:marginalFilter}
\end{align}
Thus \(F\) changes the conditional momentum statistics but not the measured
marginal or the probability law of the aligned homodyne record.

Equation~\eqref{eq:marginalFilter} preserves Gaussianity.  Writing
\begin{equation}
\rho(q,t)=\frac{1}{\sqrt{2\pi V_t}}
\exp\!\left[-\frac{(q-\bar q_t)^2}{2V_t}\right]
\label{eq:gaussianMeasuredMarginal}
\end{equation}
gives
\begin{align}
\dd\bar q_t
={}&
\left[
g_0(t)+\left(g_1(t)-\frac{\kappa}{2}\right)\bar q_t
\right]\dd t
\nonumber\\
&+
\sqrt{2\eta\kappa}
\left(V_t-\frac12\right)\dd W_t,
\label{eq:marginalMeanEvolution}\\
\dot V_t
={}&
2\left(g_1(t)-\frac{\kappa}{2}\right)V_t
+\frac{\kappa}{2}(2\nth+1)
\nonumber\\
&-
2\eta\kappa\left(V_t-\frac12\right)^2.
\label{eq:marginalVarianceEvolution}
\end{align}
The marginal enters the hierarchy below through these two scalar quantities
and the affine logarithmic derivative
\begin{equation}
\partial_q\ln\rho(q,t)
=-\frac{q-\bar q_t}{V_t}.
\label{eq:gaussianScoreClosure}
\end{equation}

\subsection{Conditional momentum hierarchy}
\label{subsec:normalizedHierarchy}

For each nonnegative integer \(n\), define
\begin{align}
J_n(q,t)&=\int_{-\infty}^{\infty}p^nW_c(q,p,t)\,\dd p,
\nonumber\\
m_n(q,t)&=\frac{J_n(q,t)}{\rho(q,t)},
\qquad m_0=1.
\label{eq:momentumSections}
\end{align}
Since the Wigner function need not be positive, \(m_n\) is a normalized
Wigner moment rather than a classical conditional probability moment.

Since \(m_n=J_n/\rho\), its stochastic change contains the update of the
numerator minus the change due solely to its normalization:
\begin{align}
\left.\dd m_n\right|_{\dd W}
={}&
\frac{1}{\rho}
\left(
\left.\dd J_n\right|_{\dd W}
-m_n\left.\dd\rho\right|_{\dd W}
\right)
\nonumber\\
={}&
\sqrt{\frac{\eta\kappa}{2}}\,
\partial_qm_n\,\dd W_t.
\label{eq:innovationCancellation}
\end{align}
Thus the common multiplicative backaction cancels under normalization, while
the derivative term acts only on \(m_n\).  Appendix
\ref{app:innovationCancellation} derives this result from It\^o's quotient
rule and evaluates the accompanying quadratic-variation drift.

Multiplying Eq.~\eqref{eq:wignerSMEClosure} by \(p^n\), integrating by parts,
and applying the normalized It\^o quotient gives
\begin{align}
\dd m_n
={}&
\Bigg\{
\frac{\kappa}{4}(2\nth+1)\partial_q^2m_n
\nonumber\\
&+
\Bigg[
-g_0(t)
-\left(g_1(t)-\frac{\kappa}{2}\right)q
\nonumber\\
&
+\frac{\kappa}{2}(2\nth+1-\eta)
\partial_q\ln\rho
-\eta\kappa(q-\bar q_t)
\Bigg]\partial_qm_n
\nonumber\\
&-
n\left(g_1(t)+\frac{\kappa}{2}\right)m_n
+\frac{\kappa}{4}(2\nth+1)n(n-1)m_{n-2}
\nonumber\\
&-
\sum_{\substack{1\leq r\leq\min(d,n)\\ r\ {\rm odd}}}
\frac{(-1)^{(r-1)/2}}{2^{r-1}}
\binom{n}{r}
F^{(r)}(q,t)m_{n-r}
\Bigg\}\dd t
\nonumber\\
&+
\sqrt{\frac{\eta\kappa}{2}}\,
\partial_qm_n\,\dd W_t,
\label{eq:generalConditionalHierarchy}
\end{align}
where \(m_k=0\) for \(k<0\).  The \(r=1\) contribution is
\(-nF'(q,t)m_{n-1}\), while the first Moyal correction is
\[
+\frac14\binom{n}{3}F'''(q,t)m_{n-3}.
\]
Thus every source in the equation for \(m_n\) comes from momentum order \(n\)
or below.
\subsection{Finite polynomial closure}
\label{subsec:finitePolynomialClosure}

Let \(\mathcal P_R\) denote the polynomials in \(q\) of degree at most \(R\).

\begin{theorem}[Finite conditional-momentum closure]
\label{thm:conditionalClosure}
Let \(F(Q,t)\) be a real polynomial of degree at most \(d\geq1\), with real
\(g_0(t)\) and \(g_1(t)\), arbitrary \(0\leq\eta\leq1\), and finite
\(\nth\geq0\).  Suppose that the initial measured-quadrature marginal is
Gaussian and that, through a requested momentum order \(N\),
\begin{equation}
m_n(q,0)\in\mathcal P_{\sigma_n},
\qquad 0\leq n\leq N,
\qquad \sigma_0=0,
\label{eq:initialPolynomialSections}
\end{equation}
where \(\sigma_n\) is a finite degree bound for the initial section \(m_n\).
Define
\begin{equation}
R_n=
\max_{0\leq k\leq n}
\left\{
\sigma_k+(n-k)(d-1)
\right\}.
\label{eq:generalDegreeBound}
\end{equation}
Then, along every conditional trajectory,
\begin{equation}
m_n(\cdot,t)\in\mathcal P_{R_n},
\qquad 0\leq n\leq N.
\label{eq:polynomialInvariance}
\end{equation}
The dynamics through order \(N\) therefore reduce exactly to a finite system
of stochastic differential equations, with no dependence on sections above
\(N\).
\end{theorem}

\begin{proof}
Consider
\begin{equation}
m_n(q,t)=\sum_{j=0}^{R_n}a_{n,j}(t)q^j,
\qquad 0\leq n\leq N.
\label{eq:generalClosureAnsatz}
\end{equation}
These forms contain the initial sections because the \(k=n\) term in
Eq.~\eqref{eq:generalDegreeBound} gives \(R_n\geq\sigma_n\).

In Eq.~\eqref{eq:generalConditionalHierarchy}, differentiation lowers degree,
while its affine transport coefficient preserves it.  Diffusion, damping,
and the stochastic term therefore remain in \(\mathcal P_{R_n}\).
Furthermore, \(R_{n-2}\leq R_n\), so the momentum-diffusion source remains in
the same space.

A Hamiltonian source has the form \(F^{(r)}m_{n-r}\), with odd \(r\geq1\).
Since \(\deg F^{(r)}\leq d-r\),
\begin{align}
\deg\!\left(F^{(r)}m_{n-r}\right)
&\leq R_{n-r}+d-r
\nonumber\\
&=
\max_{0\leq k\leq n-r}
\bigl\{\sigma_k+(n-k)(d-1)\bigr\}
\nonumber\\
&\qquad
+d(1-r)
\nonumber\\
&\leq R_n.
\label{eq:moyalDegreeEstimate}
\end{align}
The final inequality uses \(d(1-r)\leq0\) and \(n-r\leq n\).  It can be
saturated for the classical force \(r=1\), while higher Moyal terms lie
strictly below the bound. Thus the hierarchy preserves the polynomial spaces in
Eq.~\eqref{eq:generalClosureAnsatz} and introduces no section above \(m_n\).
Equating powers of \(q\) gives the closed coefficient equations.
\end{proof}

For an initial Gaussian state, the momentum distribution at fixed \(q\) has
an affine mean and \(q\)-independent variance, so \(\sigma_n\leq n\).  Hence,
for \(d\geq2\),
\begin{equation}
R_n=n(d-1).
\label{eq:gaussianDegreeBound}
\end{equation}
For an uncorrelated coherent or squeezed input, every initial section is
independent of \(q\), so \(\sigma_n=0\). 

\subsection{Observable reconstruction and complexity}
\label{subsec:observableReconstruction}

For every Weyl-ordered monomial with momentum order \(n\leq N\),
\begin{equation}
\left\langle\{Q^rP^n\}_{\mathrm W}\right\rangle_c
=
\int q^r\rho(q,t)m_n(q,t)\,\dd q.
\label{eq:observableReconstruction}
\end{equation}
Because the marginal is Gaussian and \(m_n\) is polynomial, this integral
reduces to finitely many moments of \(\rho\).  The number of coefficients
that must be evolved depends on the momentum order \(n\), not the position
order \(r\); increasing \(r\) changes only the final Gaussian
reconstruction.  Other monomial orderings follow from \([Q,P]=i\), with explicit
conversion formulas given in Appendix~\ref{app:orderingConversion}.

Substituting Eq.~\eqref{eq:generalClosureAnsatz} into
Eq.~\eqref{eq:generalConditionalHierarchy} and equating powers of \(q\)
gives the finite It\^o equations for the coefficients \(a_{n,j}(t)\).  The
number of equations through momentum order \(N\) is
\begin{equation}
C_N=\sum_{n=1}^{N}(R_n+1).
\label{eq:generalCoefficientCount}
\end{equation}
Using Eq.~\eqref{eq:gaussianDegreeBound} gives
\begin{equation}
C_{d,N}
=
\frac{N}{2}\left[(d-1)N+d+1\right].
\label{eq:coefficientCount}
\end{equation}
For fixed degree, initial polynomial bounds, and requested momentum order,
this count is independent of any Fock-space cutoff, and the work is linear in
the number of time steps.
\section{Worked example: cubic phase dynamics}
\label{sec:cubicExample}
We now illustrate how the general finite-dimensional representation
reduces to a small set of coupled stochastic equations for the cubic
phase Hamiltonian
\begin{equation}
F(Q)=\gamma Q^3,
\label{eq:cubicHamiltonian}
\end{equation}
with a coherent initial state of amplitude $\alpha$, zero thermal
occupation, and $g_0=g_1=0$. We first consider the evaluation of the
conditional quadrature expectations $\langle Q\rangle_c$ and
$\langle P\rangle_c$, before extending the construction to higher
moments.

The measured-quadrature marginal evolves independently of the
nonlinear Hamiltonian. For $\nth=0$, the variance equation
\eqref{eq:marginalVarianceEvolution} becomes
\begin{equation}
\dot V_t
=
-\kappa V_t+\frac{\kappa}{2}
-2\eta\kappa\left(V_t-\frac12\right)^2.
\end{equation}
Since a coherent state has $V_0=1/2$, the right-hand side vanishes
and the measured quadrature retains its vacuum variance,
\begin{equation}
V_t=\frac12,
\end{equation}
for every measurement efficiency $\eta$. The stochastic contribution
to the marginal mean in Eq.~\eqref{eq:marginalMeanEvolution} therefore
also vanishes, giving
\begin{equation}
\bar q_t=\bar q_0e^{-\kappa t/2},
\qquad
\bar q_0=\sqrt{2}\operatorname{Re}(\alpha).
\label{eq:cubicMarginalMean}
\end{equation}
Thus, although the conditional momentum statistics depend on the
measurement record, the measured marginal is known analytically
throughout the trajectory.

We introduce the centered coordinate $u=q-\bar q_t$. Since $\bar q_t$
is deterministic, this change of variables introduces no additional
stochastic terms. The Hamiltonian force in these coordinates is
\begin{equation}
-F'(u+\bar q_t)
=
-3\gamma(u^2+2\bar q_tu+\bar q_t^2).
\end{equation}
Substituting into Eq.~\eqref{eq:generalConditionalHierarchy} gives
\begin{align}
\dd m_n
={}&-\frac{\kappa}{2}(n+u\partial_u)m_n\,\dd t
\nonumber\\
&+\frac{\kappa}{4}\partial_u^2m_n\,\dd t
+\frac{\kappa}{4}n(n-1)m_{n-2}\,\dd t
\nonumber\\
&-3\gamma n(u+\bar q_t)^2m_{n-1}\,\dd t
\nonumber\\
&+\frac{3\gamma}{2}\binom{n}{3}m_{n-3}\,\dd t
\nonumber\\
&+\sqrt{\frac{\eta\kappa}{2}}\,
\partial_um_n\,\dd W_t.
\label{eq:cubicConditionalHierarchy}
\end{align}
The final term in the drift is the Moyal correction, which contributes
only at momentum orders $n\geq3$.

Since the Hamiltonian has degree $d=3$, the closure theorem gives
$m_n(\cdot,t)\in\mathcal P_{2n}$. In particular, the first conditional
momentum moment requires only three coefficients,
\begin{equation}
m_1(u,t)=a_0(t)+a_1(t)u+a_2(t)u^2.
\label{eq:cubicM1Expansion}
\end{equation}
Substituting this polynomial into
Eq.~\eqref{eq:cubicConditionalHierarchy} and equating the coefficients
of $u^2$, $u$, and $1$ gives
\begin{align}
\dd a_2
={}&\left(-3\gamma-\frac{3\kappa}{2}a_2\right)\dd t,
\label{eq:cubicA2}\\
\dd a_1
={}&\left(-6\gamma\bar q_t-\kappa a_1\right)\dd t
+\sqrt{2\eta\kappa}\,a_2\,\dd W_t,
\label{eq:cubicA1}\\
\dd a_0
={}&\left(-3\gamma\bar q_t^2-\frac{\kappa}{2}a_0
+\frac{\kappa}{2}a_2\right)\dd t
\nonumber\\
&+\sqrt{\frac{\eta\kappa}{2}}\,a_1\,\dd W_t.
\label{eq:cubicA0}
\end{align}
The triangular structure is particularly apparent in this example:
$a_2$ evolves deterministically, while the stochastic evolution of
$a_1$ depends on $a_2$ and that of $a_0$ depends on $a_1$.

For a coherent initial state, the position and momentum quadratures
are initially uncorrelated, so
\begin{equation}
m_1(u,0)=\sqrt{2}\operatorname{Im}(\alpha),
\end{equation}
and hence
\begin{equation}
a_0(0)=\sqrt{2}\operatorname{Im}(\alpha),
\qquad a_1(0)=a_2(0)=0.
\end{equation}

The conditional quadrature expectations are then reconstructed
from the Gaussian marginal,
\begin{align}
\langle Q\rangle_c&=\bar q_t,\\
\langle P\rangle_c
&=\int_{-\infty}^{\infty}m_1(u,t)\rho(u,t)\,\dd u
\nonumber\\
&=a_0+a_1\E_\rho[u]+a_2\E_\rho[u^2],
\end{align}
where $\rho(u,t)$ denotes the marginal in the centered coordinate
and $\E_\rho$ denotes integration against this marginal. Since
$\E_\rho[u]=0$ and $\E_\rho[u^2]=1/2$, we obtain
\begin{equation}
\boxed{
\langle P\rangle_c=a_0+\frac12a_2.
}
\label{eq:cubicMomentumReconstruction}
\end{equation}
Thus, rather than evolving the conditional density matrix or the
infinite chain of coupled observable moments
$\langle P\rangle_c,\langle\{Q,P\}\rangle_c,\ldots$, we need only
propagate three scalar coefficients.

The same construction extends to higher momentum orders.
For example, the second conditional momentum moment takes the form
\begin{equation}
m_2(u,t)=\sum_{j=0}^{4}b_j(t)u^j,
\end{equation}
with five coefficient equations obtained from
Eq.~\eqref{eq:cubicConditionalHierarchy}. Gaussian integration
then gives
\begin{align}
\langle P^2\rangle_c
&=b_0+b_2\E_\rho[u^2]+b_4\E_\rho[u^4]
\nonumber\\
&=\boxed{b_0+\frac12b_2+\frac34b_4}.
\end{align}
The complete hierarchy through second momentum order therefore
requires only $3+5=8$ coefficients, and extending to third
momentum order requires $3+5+7=15$.

More generally, for the cubic Hamiltonian, the number of
coefficients through momentum order $N$ is
\begin{equation}
C_{3,N}=\sum_{n=1}^{N}(2n+1)=N(N+2).
\end{equation}
The dimension of the hierarchy grows quadratically with the
requested momentum order, independently of the oscillator
Hilbert-space dimension.
Arbitrary powers of $Q$ enter only through the final
Gaussian integration in Eq.~\eqref{eq:observableReconstruction}.
In particular, the eight coefficients of $m_1$ and $m_2$
suffice to evaluate the mean and variance of any
polynomial-phase nullifier.
\section{Multimode extension}
\label{sec:multimode}

Let \([Q_j,P_k]=i\delta_{jk}\) for \(j,k=1,\ldots,M\), and consider
\begin{align}
H(t)={}&F(\bm Q,t)+\sum_{j=1}^{M}g_{0j}(t)P_j
\nonumber\\
&+\frac12\sum_{j,k=1}^{M}G_{jk}(t)\{Q_j,P_k\},
\nonumber\\
\bm Q={}&(Q_1,\ldots,Q_M)^T,
\label{eq:multimodeHamiltonian}
\end{align}
where \(F\) is a real polynomial of total degree at most \(d\geq1\), and
\(g_{0j}(t)\) and \(G_{jk}(t)\) are real.  The \(G\) term mixes measured
quadratures only with measured quadratures and conjugate momenta only with
conjugate momenta, preserving the separation used above.

Mode \(j\) has damping rate \(\kappa_j\), thermal occupation
\(n_{{\rm th},j}\), and homodyne efficiency \(\eta_j\).  The monitored modes
form a set \(\mathsf S\), with independent innovations \(\dd W_{j,t}\);
we set $\eta_j=0$ for $j\notin\mathsf S$.
Not every mode needs to be monitored for the hierarchy to close.
The conditional momentum sections of unmonitored modes may still depend
on the records of monitored modes through intermode interactions.

Writing \(\bm q=(q_1,\ldots,q_M)^T\) and
\(\bm p=(p_1,\ldots,p_M)^T\), define the \(Q\)-marginal
\begin{equation}
\rho(\bm q,t)
=
\int_{\mathbb R^M}
W_c(\bm q,\bm p,t)\,\dd^M\bm p.
\label{eq:multimodeMarginal}
\end{equation}
As in the single-mode case, integrating over all momenta removes \(F\).  The
remaining drift,
\begin{equation}
A_{j,t}(\bm q)
=
g_{0j}(t)+\sum_{k=1}^{M}G_{kj}(t)q_k
-\frac{\kappa_j}{2}q_j,
\label{eq:multimodeAffineDrift}
\end{equation}
is affine, and the diffusion is constant.  An initially Gaussian
\(\rho(\bm q,0)\) therefore remains Gaussian, including for correlated inputs
and partial monitoring, so every
\(\partial_{q_j}\ln\rho(\bm q,t)\) is affine in \(\bm q\).

For a momentum multi-index \(\bm n=(n_1,\ldots,n_M)\), define
\begin{align}
J_{\bm n}(\bm q,t)
&=
\int_{\mathbb R^M}\bm p^{\bm n}
W_c(\bm q,\bm p,t)\,\dd^M\bm p,
\nonumber\\
m_{\bm n}(\bm q,t)
&=
\frac{J_{\bm n}(\bm q,t)}{\rho(\bm q,t)},
\qquad m_{\bm0}=1.
\label{eq:multimodeSections}
\end{align}
Here
\(\bm p^{\bm n}=\prod_jp_j^{n_j}\),
\(|\bm n|=\sum_jn_j\),
\(\partial_{\bm q}^{\bm\alpha}=\prod_j\partial_{q_j}^{\alpha_j}\), and
\(\binom{\bm n}{\bm\alpha}=\prod_j\binom{n_j}{\alpha_j}\).
Multi-index inequalities are componentwise, and \(\bm e_j\) denotes the
\(j\)th coordinate vector.  We also write
\(\bar q_{j,t}=\langle Q_j\rangle_c\).

The same integrations by parts and normalized It\^o calculation used in the
single-mode case give
\begin{align}
\dd m_{\bm n}
={}&
\Bigg\{
\sum_{j=1}^{M}
\Bigg[
\frac{\kappa_j}{4}(2n_{{\rm th},j}+1)
\partial_{q_j}^2m_{\bm n}
\nonumber\\
&+
\Bigg(
-A_{j,t}
+\frac{\kappa_j}{2}
(2n_{{\rm th},j}+1-\eta_j)
\partial_{q_j}\ln\rho
\nonumber\\
&\hspace{28mm}
-\eta_j\kappa_j(q_j-\bar q_{j,t})
\Bigg)
\partial_{q_j}m_{\bm n}
\Bigg]
\nonumber\\
&-\frac12\sum_{j=1}^{M}\kappa_jn_jm_{\bm n}
-\sum_{j,k=1}^{M}
n_jG_{jk}(t)m_{\bm n-\bm e_j+\bm e_k}
\nonumber\\
&+\sum_{j=1}^{M}
\frac{\kappa_j}{4}(2n_{{\rm th},j}+1)
n_j(n_j-1)m_{\bm n-2\bm e_j}
\nonumber\\
&-
\sum_{\substack{
\bm0<\bm\alpha\leq\bm n\\
|\bm\alpha|\leq d,\;
|\bm\alpha|\ {\rm odd}
}}
\frac{(-1)^{(|\bm\alpha|-1)/2}}
{2^{|\bm\alpha|-1}}
\binom{\bm n}{\bm\alpha}
(\partial_{\bm q}^{\bm\alpha}F)
m_{\bm n-\bm\alpha}
\Bigg\}\dd t
\nonumber\\
&+
\sum_{j\in\mathsf S}
\sqrt{\frac{\eta_j\kappa_j}{2}}\,
\partial_{q_j}m_{\bm n}\,\dd W_{j,t},
\label{eq:multimodeConditionalHierarchy}
\end{align}
where terms with negative multi-index components are absent.  The \(G\) term
redistributes momentum factors between modes while preserving
\(|\bm n|\); the diffusion and Hamiltonian sources descend in total momentum
order.

Let \(\mathcal P_R^{(M)}\) denote the polynomials in \(\bm q\) of total degree
at most \(R\).  Fix a maximum total momentum order \(N\), and suppose
\begin{equation}
m_{\bm n}(\bm q,0)
\in\mathcal P_{\sigma_{\bm n}}^{(M)},
\qquad |\bm n|\leq N,
\qquad \sigma_{\bm0}=0,
\label{eq:multimodeInitialDegrees}
\end{equation}
where \(\sigma_{\bm n}\) is a finite degree bound for the initial section.
Because \(G\) couples sections with the same total momentum order, define
\begin{align}
\sigma_s&=\max_{|\bm n|=s}\sigma_{\bm n},
\nonumber\\
R_s&=\max_{0\leq r\leq s}
\left\{\sigma_r+(s-r)(d-1)\right\}.
\label{eq:multimodeDegreeBound}
\end{align}

\begin{theorem}[Multimode finite conditional-momentum closure]
\label{thm:multimodeClosure}
If the initial \(Q\)-marginal is Gaussian and
Eq.~\eqref{eq:multimodeInitialDegrees} holds, then
\begin{equation}
m_{\bm n}(\cdot,t)
\in\mathcal P_{R_{|\bm n|}}^{(M)},
\qquad |\bm n|\leq N,
\label{eq:multimodePolynomialInvariance}
\end{equation}
along every conditional trajectory.  The dynamics through total momentum
order \(N\) therefore form an exact finite stochastic system, with no
dependence on higher-order sections.
\end{theorem}

\begin{proof}
For every \(|\bm n|=s\leq N\), consider
\begin{equation}
m_{\bm n}(\bm q,t)
=
\sum_{|\bm\beta|\leq R_s}
a_{\bm n,\bm\beta}(t)\bm q^{\bm\beta}.
\label{eq:multimodeClosureAnsatz}
\end{equation}
These forms contain the initial sections because \(R_s\geq\sigma_s\).

In Eq.~\eqref{eq:multimodeConditionalHierarchy}, differentiation lowers total
degree, while the affine transport coefficients preserve it.  Diffusion,
damping, and the stochastic terms therefore remain in
\(\mathcal P_{R_s}^{(M)}\).  The \(G\) term couples only sections in the same
order-\(s\) block, while \(R_{s-2}\leq R_s\) places the momentum-diffusion
sources in the same space.

A Hamiltonian source indexed by \(\bm\alpha\), with
\(a=|\bm\alpha|\), has the form
\((\partial_{\bm q}^{\bm\alpha}F)m_{\bm n-\bm\alpha}\).  Since
\(\deg(\partial_{\bm q}^{\bm\alpha}F)\leq d-a\),
\begin{align}
\deg\!\left[
(\partial_{\bm q}^{\bm\alpha}F)
m_{\bm n-\bm\alpha}
\right]
&\leq R_{s-a}+d-a
\nonumber\\
&=
\max_{0\leq r\leq s-a}
\left\{\sigma_r+(s-r)(d-1)\right\}
\nonumber\\
&\qquad+d(1-a)
\leq R_s.
\label{eq:multimodeDegreeEstimate}
\end{align}
The final inequality follows from \(d(1-a)\leq0\) and \(s-a\leq s\).
Thus the hierarchy preserves the finite polynomial blocks in
Eq.~\eqref{eq:multimodeClosureAnsatz} and introduces no section above total
order \(s\).  Equating multivariate monomials gives the closed coefficient
equations.
\end{proof}

For a correlated Gaussian input, the conditional momentum means are affine in
\(\bm q\) and their covariance is independent of \(\bm q\), so
\(\sigma_s\leq s\).  Hence, for \(d\geq2\),
\begin{equation}
R_s=s(d-1).
\label{eq:multimodeGaussianDegreeBound}
\end{equation}
The number of coefficients through total momentum order \(N\) is then
\begin{equation}
C_{M,d,N}
=
\sum_{s=1}^{N}
\binom{M+s-1}{s}
\binom{M+s(d-1)}{M}.
\label{eq:multimodeCoefficientCount}
\end{equation}
For fixed \(M,d,N\), the work remains linear in trajectory length and
independent of all Fock cutoffs.  In particular, interacting polynomials such
as \(F(\bm Q)=\lambda Q_1Q_2Q_3\) are covered by the same construction and
do not reduce to independent single-mode filters.

\section{Computational complexity and scaling}
\label{sec:computationalComplexity}

The coefficient count $C_{M,d,N}$ in
Eq.~\eqref{eq:multimodeCoefficientCount} provides an upper bound
for Gaussian inputs on the number of coefficients required to describe the conditional
dynamics through total momentum order $N$ for an $M$-mode system
with Hamiltonian degree at most $d$. This bound allows every
momentum section to depend on every position coordinate, whereas
in practice the Hamiltonian structure and the choice of observables
often permit substantially smaller representations. Including the
$M$ means and $M(M+1)/2$ covariance entries of the Gaussian
measured marginal, the general representation has dimension
\begin{equation}
D_{M,d,N}
=
C_{M,d,N}+\frac{M(M+3)}{2}.
\label{eq:totalMultimodeDimension}
\end{equation}
For deterministic system parameters, the marginal covariance
evolves independently of the measurement record and may be
computed in advance.

Using the asymptotic identity
\begin{equation}
\binom{x+a}{b}
\sim
\frac{x^b}{b!},
\qquad x\rightarrow\infty,
\label{eq:binomialIdentity}
\end{equation}
for fixed $a$ and $b$, the summand in
Eq.~\eqref{eq:multimodeCoefficientCount} satisfies
\begin{equation}
\binom{M+s-1}{s}
\binom{M+s(d-1)}{M}
\sim
\frac{M^{sd}}{s![s(d-1)]!}.
\end{equation}
For fixed Hamiltonian degree $d\geq2$ and momentum order $N$,
the highest momentum-order block therefore dominates, giving
\begin{equation}
\boxed{
C_{M,d,N}
\sim
\frac{M^{Nd}}{N![N(d-1)]!}.
}
\label{eq:modeNumberScaling}
\end{equation}
Similarly, applying Eq.~\eqref{eq:binomialIdentity} in the
limits of increasing momentum order and Hamiltonian degree gives
\begin{align}
C_{M,d,N}
&\sim
\frac{(d-1)^M}{2(M!)^2}N^{2M},
\label{eq:momentumOrderScaling}\\
C_{M,d,N}
&=O(d^M),
\label{eq:hamiltonianDegreeScaling}
\end{align}
respectively, with the remaining parameters held fixed.
The general representation thus grows polynomially in each
of its three parameters individually.

The polynomial scaling with mode number is particularly
significant in comparison with direct density-matrix
propagation. For a local Fock cutoff $N_{\rm F}$, the
$M$-mode Hilbert space has dimension $N_{\rm F}^M$, giving
the representation sizes
\begin{align}
S_{\rm SME}&=\Theta(N_{\rm F}^{2M}),\\
S_{\rm CMM}&=O(M^{Nd}),
\label{eq:storageScalingComparison}
\end{align}
at fixed $d\geq2$ and $N\geq1$. Unlike the Fock-space
representation, the conditional-momentum representation
introduces no Hilbert-space truncation error and requires
no increase in dimension with increasing oscillator
occupation. Although the arithmetic cost per time step
depends on the sparsity and coupling structure of the
coefficient equations, it remains polynomial in $M$ for
fixed $d$ and $N$, providing a generic simulation advantage for these Hamiltonians.

This upper bound can be substantially improved for spatially
local polynomial-phase Hamiltonians. For nearest-neighbor
interactions on a square lattice as in Eq. \ref{eq:latticeHamiltonian},
\begin{equation}
[F(\bm Q),P_i]=i\partial_{Q_i}F(\bm Q)
\end{equation}
depends only on mode $i$ and its neighbors. Since $[F,Q_i]=0$,
these interactions do not propagate position dependence to more
distant modes. For product coherent inputs, independent vacuum
damping, and aligned local homodyne monitoring, the remaining
hierarchy terms preserve this locality. Each first-order momentum
section therefore requires only a bounded number of coefficients.
Similarly, the identity
\begin{equation}
[F,P_iP_j]=[F,P_i]P_j+P_i[F,P_j]
\end{equation}
restricts the second-order section to the combined neighborhoods
of modes $i$ and $j$. For disjoint neighborhoods, the sections
factorize as $m_{\bm e_i+\bm e_j}=m_{\bm e_i}m_{\bm e_j}$, as
shown in Appendix~\ref{app:collectiveLattice}. Additional
second-order coefficients are needed only for overlapping
neighborhoods. Since each mode has a bounded number of such
neighbors, the representation through second momentum order
scales as $O(M)$. For cubic interactions, this improves the
general $O(M^6)$ coefficient bound. We demonstrate this scaling
numerically in Sec.~\ref{subsec:latticeScaling}.

\section{Numerical verification of the finite hierarchy}
\label{sec:closureNumerics}

Although the closure established above is exact in continuous time, we
verify its numerical implementation against direct Fock-space quantum
trajectories.  These comparisons test shared-record evolution, observable
reconstruction, the Moyal correction, multimode coupling, numerical
convergence, and computational cost.

All Fock-space references are mixed-state trajectories generated with
QuTiP's \texttt{smesolve} \cite{qutip5}.  For efficiency
\(\eta<1\), the observed channel \(\sqrt{\eta\kappa}\,a\) is supplied as a
stochastic collapse operator and the unobserved channel
\(\sqrt{(1-\eta)\kappa}\,a\) as an ordinary collapse operator.  We use the
positivity-preserving Rouchon update \cite{rouchon2015} as a stable numerical
integrator of the SME; it does not define a different physical model.

For a homodyne current binned over intervals of width \(\Delta t\), the
left-endpoint innovation is
\begin{equation}
\Delta W_k=
\left[J_k-\sqrt{2\eta\kappa}\,
\langle Q\rangle_{c,k}\right]\Delta t.
\label{eq:numericalInnovationExtraction}
\end{equation}
The coefficient equations are driven by this same \(\Delta W_k\) and receive
no Fock-space state or reference moment. This is therefore the same causal reconstruction that would be applied to an experimental current.

The coefficient SDEs are integrated with Euler--Maruyama.  Discrepancies can
therefore arise from the Fock cutoff and internal Rouchon step, the hierarchy time step, and the left-endpoint reconstruction in
Eq.~\eqref{eq:numericalInnovationExtraction}.  Shared records remove
trajectory-to-trajectory variation, while cutoff scans and time step
refinement isolate these errors.  We focus on observables containing \(P\),
because \(Q\)-only moments follow directly from the closed Gaussian marginal
and do not test the conditional-momentum hierarchy.

\subsection{Trajectory-level agreement}
\label{subsec:hierarchyVerification}
\begin{figure*}[t]
\centering
\includegraphics[width=\textwidth]{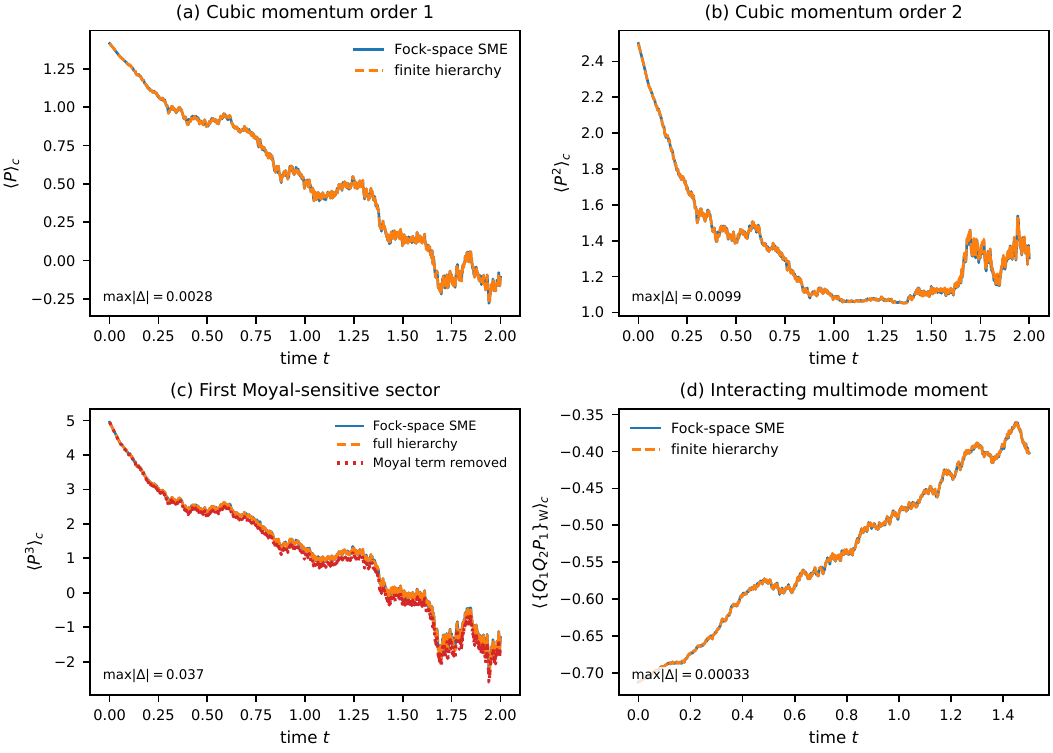}
\caption{%
Shared-record validation against mixed-state Fock-space SMEs at \(\eta=0.8\).
(a)--(c) Conditional momentum moments for a cubic Hamiltonian.  The finite
hierarchy follows the SME through \(P^3\), whereas deleting the first Moyal
source visibly fails in panel (c).  (d) Agreement for a genuinely bivariate
moment under \(0.10Q_1^2Q_2\).  Annotations give the maximum absolute
trajectory discrepancies.}
\label{fig:exactHierarchyValidation}
\end{figure*}

Figure~\ref{fig:exactHierarchyValidation} compares the finite hierarchy with
the density-matrix SME at \(\eta=0.8\).  The single-mode example uses
\(F(Q)=0.15Q^3\), \(\kappa=0.5\), \(\alpha=1+i\), \(T=2\), Fock cutoff 100,
record spacing \(5\times10^{-4}\), and internal Rouchon step
\(10^{-4}\).  The root-mean-square discrepancies in
\(\langle P\rangle_c\), \(\langle P^2\rangle_c\), and
\(\langle P^3\rangle_c\) are \(9.35\times10^{-4}\),
\(2.73\times10^{-3}\), and \(9.70\times10^{-3}\), respectively.  The
population in the highest five Fock levels is \(4.05\times10^{-9}\).

The third-momentum sector tests the first genuinely quantum term in the
hierarchy.  For \(F(Q)=\gamma Q^3\), its equation contains the Moyal source
\begin{equation}
\frac{3\gamma}{2}m_0.
\label{eq:cubicFirstMoyalSource}
\end{equation}
Removing only this term increases the \(P^3\) discrepancy to \(1.65\times
10^{-1}\).  The agreement of the complete hierarchy is therefore not
explained by a classical drift--diffusion model.

The multimode test uses \(F(Q_1,Q_2)=0.10Q_1^2Q_2\), \(\kappa=0.5\),
\(\alpha_1=0.8+0.7i\), \(\alpha_2=-0.45+0.65i\), cutoff 16 per mode,
\(T=1.5\), and the same \(\eta=0.8\).  The mixed moment
\(\langle\{Q_1Q_2P_1\}_{\rm W}\rangle_c\), which cannot be reconstructed
from independent single-mode filters, has root-mean-square discrepancy
\(1.48\times10^{-4}\). 
Euler--Maruyama has strong order one half \cite{Higham2001}, so the
$\Delta t=5\times10^{-4}$ time step corresponds to a nominal
$\sqrt{\Delta t}\simeq2.2\times10^{-2}$ discretization scale before
problem-dependent prefactors.  The observed discrepancies, ranging from
$1.48\times10^{-4}$ to $9.70\times10^{-3}$, are consistent with this scale
and decrease under refinement in Fig.~\ref{fig:convergenceScaling}(a).

\subsection{Convergence and cutoff-independent cost}
\label{subsec:convergenceScaling}

\begin{figure*}[t]
\centering
\includegraphics[width=\textwidth]
{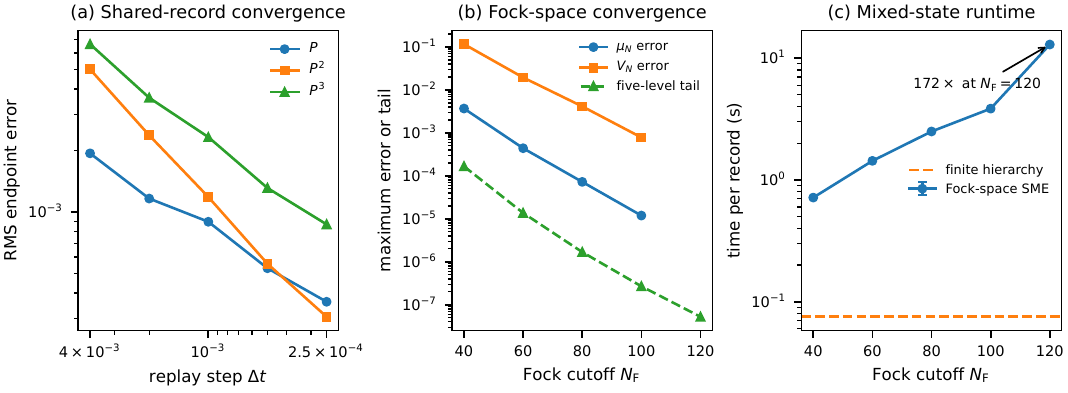}
\caption{%
Accuracy and cost for mixed-state SME benchmarks.  (a) Shared-record endpoint
errors as the hierarchy time step is refined.  (b) Maximum-in-time
nullifier-mean and nullifier-variance deviations from cutoff 120, together
with the final population in the highest five retained Fock levels, shown on
the same logarithmic axis for convenience.   The population is an independent
diagnostic of occupation near the truncation boundary. The cutoff-120 error points are zero by definition and are omitted from
the logarithmic plot. (c) Median time for one density-matrix SME record versus Fock cutoff;
the dashed line is the finite-hierarchy time.  Error bars in (c) show the
interquartile range of three runs.  Absolute runtimes are machine dependent.}
\label{fig:convergenceScaling}
\end{figure*}

Figure~\ref{fig:convergenceScaling}(a) repeats the single-mode comparison for
12 mixed-state SME records while refining only the time step.  At
\(\Delta t=2.5\times10^{-4}\), the endpoint root-mean-square errors are
\(3.62\times10^{-4}\), \(3.06\times10^{-4}\), and
\(8.68\times10^{-4}\) for \(P\), \(P^2\), and \(P^3\), respectively, and
all three decrease monotonically under step refinement.

Panels (b) and (c) use a more demanding four-decibel squeezed input with
\(F(Q)=0.35Q^3\), \(\kappa T=0.75\), and \(\eta=0.8\), similar to the
cubic nullifier considered below.  Relative to cutoff 120, the maximum
nullifier-variance deviation decreases from \(1.17\times10^{-1}\) at cutoff
40 to \(7.85\times10^{-4}\) at cutoff 100.  As an independent truncation
diagnostic, panel (b) also shows the final population
\[
p_{\mathrm{tail}}(T)
=
\sum_{n=N_F-5}^{N_F-1}
\langle n|\rho_c(T)|n\rangle
\]
in the highest five retained Fock levels.  Its maximum across the three runs
falls from \(1.67\times10^{-4}\) at cutoff 40 to \(5.25\times10^{-8}\) at
cutoff 120, showing that the artificial boundary becomes negligibly
occupied.

The same runs provide the timing comparison.  At cutoff 120 the median SME
time is \(12.94\,\mathrm{s}\), while the cutoff-independent hierarchy takes
\(0.075\,\mathrm{s}\) per record when averaged across the scan, a factor of
approximately $172$ for this workload. Although absolute timings depend on hardware and implementation, the
structural result is that the hierarchy size is independent of the physical
parameter values and Fock-space cutoff.
\subsection{Scaling in an interacting multimode network}
\label{subsec:latticeScaling}

\begin{figure}[t]
\centering
\includegraphics[width=\columnwidth]{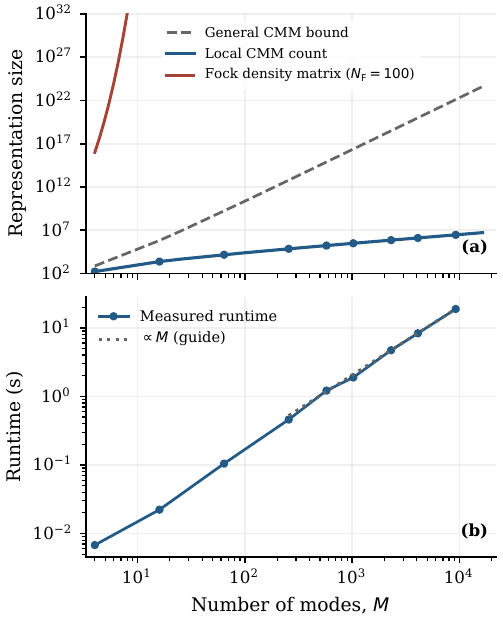}
\caption{Computational scaling for the square-lattice Hamiltonian
in Eq.~\eqref{eq:latticeHamiltonian}.
(a) Local momentum coefficient count compared with the general
polynomial bound $C_{M,3,2}$ and the dense Fock-space density-matrix
entry count $N_{\rm F}^{2M}$ for $N_{\rm F}=100$. All three curves
are analytical counts; markers on the local curve denote simulated
lattice sizes. The local coefficient count asymptotically approaches
$332M$, as derived in Appendix~\ref{app:collectiveLattice}.
(b) Measured runtime for one trajectory of 1000 time steps,
including initialization and sparse endpoint covariance
reconstruction. Initial compilation is excluded. The dotted line
indicates linear scaling for reference.}
\label{fig:latticeScaling}
\end{figure}

We consider an open square lattice of $M=L^2$ bosonic modes
with Hamiltonian
\begin{equation}
F(\bm Q)=\gamma\sum_iQ_i^3+
 g\sum_{\langle i,j\rangle}
 \left(\frac{Q_i+Q_j}{\sqrt{2}}\right)^3,
\label{eq:latticeHamiltonian}
\end{equation}
where the second sum is over nearest-neighbor pairs. This Hamiltonian extends the engineered single-mode cubic
phase interaction of Ref.~\cite{Eriksson2024} to cubic
interactions acting on pairs of neighboring modes, while
preserving compatibility with aligned homodyne monitoring. We use product
coherent initial states, $\gamma=0.10$, $g=0.08$, independent
vacuum damping at rate $\kappa=0.5$, and local $Q_i$ homodyne
monitoring with efficiency $\eta=0.8$ at every site.
We evolve the local momentum sections through second order using
the coefficient representation in Appendix~\ref{app:collectiveLattice}
to reconstruct all conditional momentum means and covariances.

Figure~\ref{fig:latticeScaling} illustrates the reduction from
the general $O(M^6)$ coefficient bound to $O(M)$ for this lattice.
For a $96\times96$ lattice ($M=9216$), the local representation
contains approximately $2.98\times10^6$ coefficients, and a single
trajectory of duration $T=2$ with $\Delta t=2\times10^{-3}$ takes
$18.8\,\mathrm{s}$ in the compiled implementation. The covariance
matrix is reconstructed and stored sparsely. The Fock-space curve
is a representation-size comparison, not a measured runtime or an
accuracy-matched truncation. This example illustrates the computational savings available
from the particular interaction structure. Direct SME
validation at these system sizes is computationally
intractable due to the exponential growth of the
Fock-space representation with mode number.
\section{Optimal record-based recentering of polynomial nullifiers}
\label{sec:feedback}

We now use the finite hierarchy to extract and optimally recenter polynomial
nullifiers from the measurement record.

\subsection{Conditional statistics and optimal displacement}
\label{subsec:nullifierStatistics}

Generalizing the polynomial-phase nullifier in
Eq.~\eqref{eq:nullifierDefinition}, let \(N_t=P+G(Q,t)\), where \(G\) is a
real polynomial.  Applying the reconstruction formula
Eq.~\eqref{eq:observableReconstruction} to its Weyl-ordered moments gives
\begin{align}
\mu_t
&=
\int\rho(q,t)\,[m_1(q,t)+G(q,t)]\,\dd q,
\label{eq:nullifierMeanFromSections}\\
s_t
&=
\int\rho(q,t)\,
[m_2(q,t)+2G(q,t)m_1(q,t)+G(q,t)^2]\,\dd q.
\label{eq:nullifierSecondMomentFromSections}
\end{align}
Thus \(V_{N,c}(t)=s_t-\mu_t^2\) is determined exactly by \(m_1\) and
\(m_2\), independently of the degree of \(G\).  Driven causally by the
innovation in Eq.~\eqref{eq:feedbackInnovationFromRecord}, the Gaussian
marginal and finite coefficients determine both quantities from the record
up to time \(t\).

Let \(Y\) denote the available record and
\(D_P(f)=\exp[-if(Y)Q]\).  Since
\(D_P(f)^\dagger N_tD_P(f)=N_t-f(Y)\), the conditional error decomposes as
\begin{align}
\left\langle[N_t-f(Y)]^2\right\rangle_c
&=
V_{N,c}(t)+[\mu_t-f(Y)]^2,
\label{eq:conditionalRiskDecomposition}\\
f_{\rm opt}(Y)&=\mu_t.
\label{eq:optimalRecordDisplacement}
\end{align}
The optimal correction therefore requires neither a fitted regression model
nor a Gaussian approximation to the conditional state.

Equivalently, the law of total variance gives
\begin{align}
\operatorname{Var}_{\rm unc}(N_t)
&=
\E[V_{N,c}(t)]+\operatorname{Var}(\mu_t),
\label{eq:totalNullifierVariance}\\
\operatorname{Var}_{\rm corr}(N_t)
&=
\E[V_{N,c}(t)]
\nonumber\\
&=
\operatorname{Var}_{\rm unc}(N_t)-\operatorname{Var}(\mu_t).
\label{eq:feedbackVarianceReduction}
\end{align}
The removable contribution is therefore exactly the between-record variance
of the conditional mean.  Recentring does not reduce the variance within a
trajectory or recover information lost through attenuation or unobserved
noise.

The same argument applies to the multimode hierarchy of
Sec.~\ref{sec:multimode}.  For a polynomial phase \(\Phi(\bm Q)\), define
\begin{equation}
N_j=P_j+\partial_{Q_j}\Phi(\bm Q),
\qquad j=1,\ldots,M.
\label{eq:multimodeNullifiers}
\end{equation}
The sectors \(|\bm n|\leq2\) determine the conditional mean vector
\(\bm\mu_t\) and covariance \(\bm\Sigma_c\).  Because the corresponding
momentum displacements commute, applying
\(\exp[-i\sum_j\mu_{j,t}Q_j]\) gives
\begin{equation}
\bm\Sigma_{\rm corr}
=
\E[\bm\Sigma_c]
=
\bm\Sigma_{\rm unc}-\operatorname{Cov}(\bm\mu_t).
\label{eq:multimodeFeedbackCovariance}
\end{equation}
The interacting example considered below is the elementary triadic
continuous-variable hypergraph resource \cite{Moore2019},
\begin{equation}
U=\exp(-i\lambda Q_1Q_2Q_3),
\qquad
N_j=P_j+\lambda\!\prod_{k\neq j}Q_k.
\label{eq:triadicHypergraphResource}
\end{equation}
Nonlinear nullifiers of these resources and their robustness to loss and
thermalization were recently analyzed in Ref.~\cite{Ravikumar2026}; here we
instead ask how much of the variance generated during monitored loss can be
removed using the measurement record.

\subsection{Nonlinear-squeezing criterion}
\label{subsec:nonlinearSqueezingCriterion}

For the monomial nullifier
\begin{equation}
N_{r,z}=P+zQ^r,
\qquad r\geq2,
\label{eq:monomialNullifier}
\end{equation}
we normalize its variance by the optimized Gaussian bound,
\begin{equation}
\xi_r(z)
=
\frac{\operatorname{Var}(N_{r,z})}{B_r(z)},
\qquad
B_r(z)
=
\min_{\psi_{\rm G}}
\operatorname{Var}_{\psi_{\rm G}}(N_{r,z}).
\label{eq:nonlinearSqueezingRatio}
\end{equation}
Thus \(\xi_r(z)<1\) certifies nonlinear squeezing below the minimum
attainable by a Gaussian state
\cite{Marek2011,Konno2021,Kala2022}.

Appendix~\ref{app:gaussianNullifierBound} explicitly performs the
optimization over Gaussian displacement, covariance, and quadrature
variance.  Defining
\begin{align}
C_r
&=
(2r-1)!!
-\delta_{r\,\mathrm{even}}[(r-1)!!]^2
\nonumber\\
&\quad
-r^2\delta_{r\,\mathrm{odd}}[(r-2)!!]^2,
\label{eq:gaussianBoundCoefficient}\\
v_*(r,z)
&=
(4rC_rz^2)^{-1/(r+1)},
\label{eq:optimalGaussianVariance}
\end{align}
the calculation yields
\begin{equation}
B_r(z)=\frac{1}{4v_*}+C_rz^2v_*^r.
\label{eq:gaussianNullifierBound}
\end{equation}
Here \(n!!=n(n-2)(n-4)\cdots\), with \(0!!=(-1)!!=1\), while
\(\delta_{r\,\mathrm{even}}\) and \(\delta_{r\,\mathrm{odd}}\) select the
indicated parity of \(r\).

The uncorrected and optimally recentered witnesses used below are therefore
\begin{align}
\xi_{\rm unc}
&=
\frac{\E[V_{N,c}(T)]+\operatorname{Var}(\mu_T)}{B_r(z)},
\nonumber\\
\xi_{\rm corr}
&=
\frac{\E[V_{N,c}(T)]}{B_r(z)}.
\label{eq:numericalWitnessRatios}
\end{align}
The first discards the measurement record, whereas the second applies the
optimal displacement before averaging over trajectories.
\section{Numerical recovery of nonlinear squeezing}
\label{sec:applicationNumerics}
Having derived the optimal record-conditioned displacement, we now quantify
the nonlinear squeezing it recovers, first in favorable regimes and then in
a loss-limited projection based on current experimental parameters.

The single-mode simulations begin with a momentum-squeezed Gaussian state and
evolve under \(F(Q)=\gamma Q^d\), damping, and aligned homodyne monitoring.
At time \(T\), the target nullifier has \(r=d-1\) and \(z=d\gamma T\).
Only \(m_1\) and \(m_2\) are propagated, from which each trajectory supplies
\(\mu_T\) and \(V_{N,c}(T)\).  Confidence intervals resample complete
trajectories, and every reported conclusion is checked under time-step
refinement.  Production runs use a step for which no sampled conditional
variance is negative.
\subsection{Single-mode recovery}
\label{subsec:cubicRecovery}
\begin{figure*}[!t]
\centering
\includegraphics[width=\textwidth]{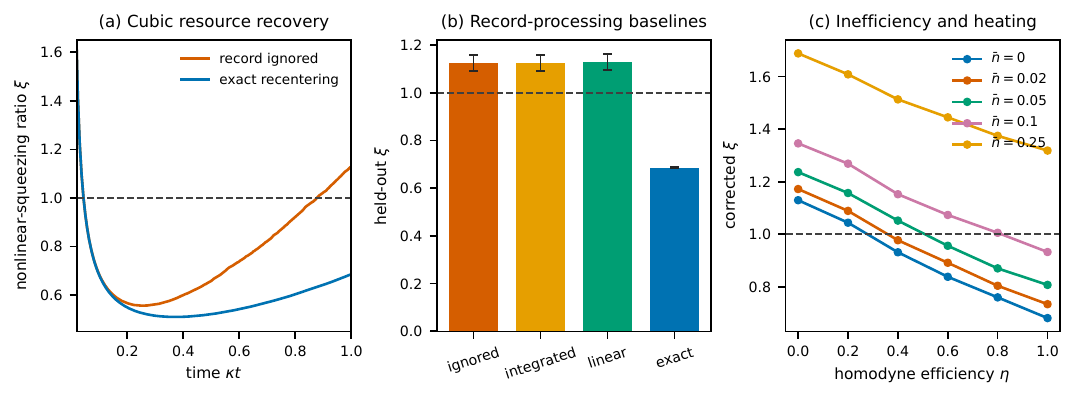}
\caption{%
Record-assisted recovery for a favorable cubic resource.  (a) Nonlinear-
squeezing ratio during the phase evolution with the record ignored and after
exact recentering; \(\xi=1\) is the optimized Gaussian threshold.  (b)
Held-out record-processing baselines with 95\% trajectory-bootstrap
intervals.  (c) Corrected endpoint ratio versus homodyne efficiency and
thermal occupation.  Unless varied in (c), \(z=1\), \(\kappa T=1\), the input
squeezing is 7 dB, \(\eta=1\), and \(\nth=0\).}
\label{fig:singleModeRecovery}
\end{figure*}

The first application is intentionally loss dominated, showing the size of
the recoverable effect rather than modeling a specific device.  We use a
cubic phase evolution with \(z=1\), \(\kappa T=1\), seven decibels of initial
momentum squeezing, \(\eta=1\), and \(\nth=0\).  An ensemble of
\(2.0\times10^4\) trajectories is propagated with
\(\kappa\Delta t=2.5\times10^{-4}\).  At the endpoint,
\begin{equation}
\xi_{\rm unc}=1.12768,
\qquad
\xi_{\rm corr}=0.68530.
\label{eq:cubicHeadlineRatios}
\end{equation}
Exact recentering removes \(39.2\%\) of the total nullifier variance and
restores a nonlinear-squeezing witness that is lost when the record is
ignored.

Figure~\ref{fig:singleModeRecovery}(b) compares the exact filter with two
held-out baselines trained on the same binned currents.  The ratios are
\(1.1257\), \(1.1262\), \(1.1294\), and \(0.6861\) for no correction, an
integrated-current gain, a 40-bin linear temporal filter, and the exact
conditional mean.  The corresponding 95\% intervals for the ignored and
exact results are \([1.0919,1.1582]\) and \([0.6842,0.6883]\).  The baseline construction and held-out evaluation are described in
Appendix~\ref{app:regressionBaselines}. The linear baselines do not
recover the witness.  This does not exclude a sufficiently
expressive nonlinear learned estimator, but the hierarchy supplies the
required statistic without training.

The correction degrades continuously as information is lost.  On the grid in
panel (c), \(\xi_{\rm corr}<1\) is retained down to \(\eta=0.4\) for
\(\nth\leq0.02\), to \(\eta=0.6\) for \(\nth=0.05\), and only at
\(\eta=1\) for \(\nth=0.1\).  No crossing remains for \(\nth=0.25\).

\subsection{Higher-degree and multimode coupling}
\label{subsec:generalityRecovery}
\begin{figure*}[!t]
\centering
\includegraphics[width=\textwidth]{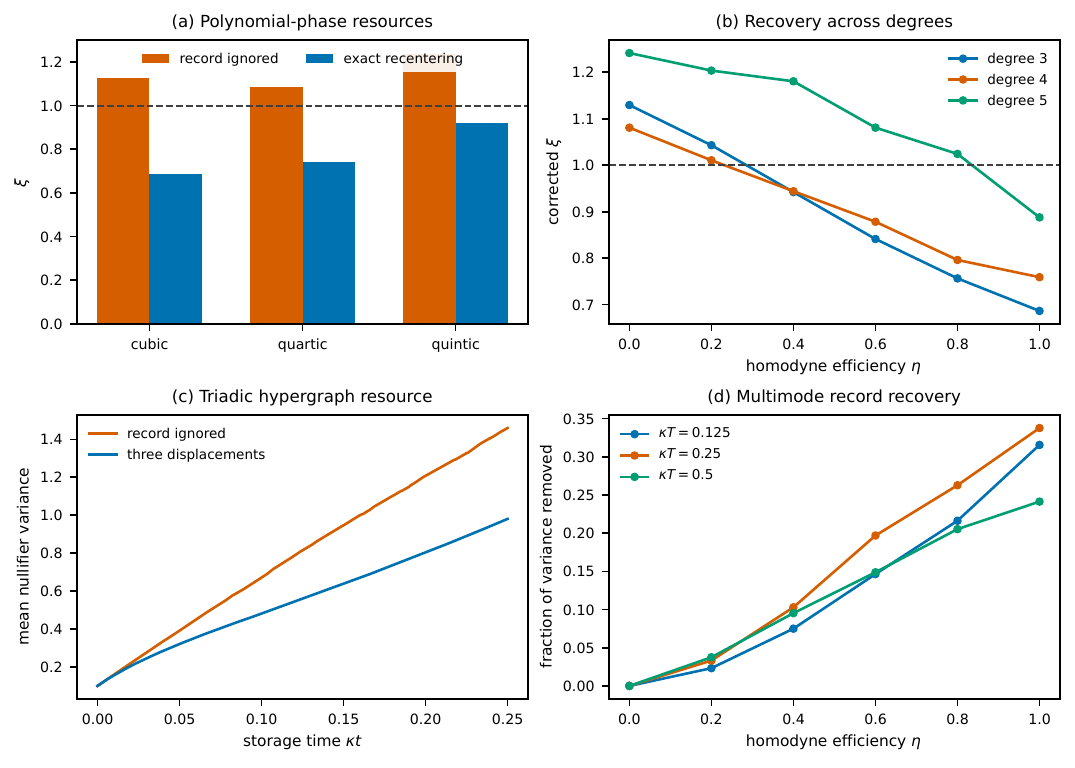}
\caption{%
Polynomial and multimode generality.  (a) Record-ignored and recentered
ratios for representative cubic, quartic, and quintic resources.  (b)
Corrected ratio versus efficiency for the same points.  (c) Mean nullifier
variance for the interacting resource \(\exp(-i\lambda Q_1Q_2Q_3)\), with
the records ignored and after three conditional displacements.  (d) Fraction
of multimode variance removed versus efficiency at three storage times.}
\label{fig:degreeMultimodeGenerality}
\end{figure*}

Figure~\ref{fig:degreeMultimodeGenerality}(a) applies the same calculation to
three polynomial degrees.  The pairs
\((\xi_{\rm unc},\xi_{\rm corr})\) are
\begin{equation}
\begin{aligned}
d=3:&\quad(1.12768,0.68530),\\
d=4:&\quad(1.08758,0.74301),\\
d=5:&\quad(1.23618,0.91938).
\end{aligned}
\label{eq:degreeCrossings}
\end{equation}
The corresponding parameters are cubic
\((z=1,\kappa T=1,7\,\mathrm{dB})\), quartic
\((z=0.5,\kappa T=0.25,5\,\mathrm{dB})\), and quintic
\((z=0.1,\kappa T=0.5,5\,\mathrm{dB})\).  Panel (b) shows that the
efficiency required to cross the Gaussian bound generally increases with degree for
these representative points.

Panels (c) and (d) evaluate the triadic resource in
Eq.~\eqref{eq:triadicHypergraphResource}.  With seven-decibel inputs,
\(\lambda=1.5\), \(\kappa T=0.25\), and unit efficiency, three simultaneous
conditional displacements remove \(33.4\%\), \(33.1\%\), and \(31.9\%\) of
the respective nullifier variances.  This uses the multimode hierarchy
directly and cannot be reduced to three independent single-mode filters.

\subsection{Loss-limited hardware projection}
\label{subsec:experimentProjection}

\begin{figure*}[!t]
\centering
\includegraphics[width=\textwidth]{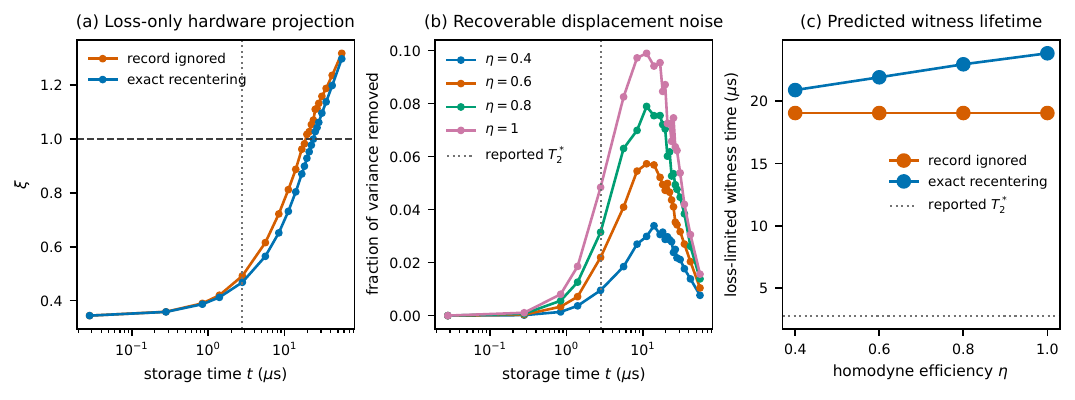}
\caption{%
Loss-limited projection using parameters from
Ref.~\cite{Eriksson2024}.  (a) Nonlinear-squeezing ratio with the record
ignored and after exact recentering at unit efficiency.  (b) Recoverable
fraction versus storage time and efficiency.  (c) Time at which the ratio
crosses \(\xi=1\).  Dotted lines mark the reported \(T_2^*\); number
dephasing is not included.}
\label{fig:hardwareProjection}
\end{figure*}

To project performance at current hardware scales,
Figure~\ref{fig:hardwareProjection} uses parameters reported for a
superconducting cubic-phase resource in Ref.~\cite{Eriksson2024}: 5.3 dB of
squeezing, cubicity 0.11, corresponding to an initial nullifier coefficient
\(z_0=0.33\), \(\nth=0.024\), \(T_1=28\,\mu\mathrm{s}\), and
\(T_2^*=2.8\,\mu\mathrm{s}\).  At time \(t\), we evaluate the
attenuation-matched nullifier \(N_t=P+z(t)Q^2\), where
\(z(t)=z_0e^{t/(2T_1)}\); this compensates for the relative attenuation of
\(P\) and \(Q^2\).  We additionally assume continuous phase-aligned homodyne
monitoring throughout the evolution, followed by the record-conditioned
displacement; this complete loop was not implemented in that experiment.

Here \(\eta\) is the total monitoring efficiency, combining the fraction of
the oscillator decay entering the monitored output with the efficiency with
which that output is detected.  Thus \(\eta=1\) represents an ideal
loss-limited upper bound.

The model includes amplitude damping, thermal diffusion, inefficient
monitoring, and homodyne backaction, but omits the additional phase
randomization represented by the reported \(T_2^*\), which was dominated by
low-frequency oscillator-frequency noise.  A Markovian approximation to such
noise would introduce a number-dephasing term proportional to
\(\mathcal D[a^\dagger a]\rho\), which lies outside the closed hierarchy
considered here.  Figure~\ref{fig:hardwareProjection} is therefore a
loss-limited projection rather than a simulation of the complete reported
device.

At the reported \(T_2^*\), the unit-efficiency model gives
\(\xi_{\rm unc}=0.4917\) and \(\xi_{\rm corr}=0.4679\), a \(4.8\%\)
variance reduction.  The recoverable fraction reaches \(9.9\%\) near
\(11.2\,\mu\mathrm{s}\).  Defining the loss-limited witness lifetime by the
first crossing of \(\xi=1\), recentering delays it from
\(18.9\,\mu\mathrm{s}\) to \(23.8\,\mu\mathrm{s}\) at unit efficiency, an
extension of \(4.9\,\mu\mathrm{s}\), or approximately \(26\%\).  The
corrected crossing times at \(\eta=0.4,0.6,0.8,1\) are
\(20.9,21.9,23.0,\) and \(23.8\,\mu\mathrm{s}\), respectively.

Since the reported \(T_2^*\) is shorter than the peak-recovery timescale,
the projection identifies improved phase coherence as the key hardware
advance for realizing the full benefit of record-assisted recentering.

\section{Discussion}
\label{sec:discussion}

The shared-record comparisons show that the finite hierarchy reproduces the
same conditional observables as the mixed-state SME, including Moyal
corrections and correlations generated by multimode interactions.  Rather than
approximating the conditional state, the hierarchy evolves only the
coefficients needed for those observables.  This is the source of its
computational advantage.  The measured speedup is hardware and implementation specific,
but the independence of Fock-space cutoff is structural.

The nullifier application gives this reduction a direct operational role.
The hierarchy extracts the exact conditional mean and variance from the
homodyne current, and the conditional mean determines the
minimum-mean-square displacement.  The tested linear record and integrated processors do not
recover the same information.  This displacement removes the between-record
variance \(\operatorname{Var}(\mu_T)\) while leaving the average conditional
variance unchanged.

The hardware projection shows how the attainable benefit
depends on the fraction of the oscillator decay that is
monitored. Within the noise model considered here, the
filter determines the optimal displacement at arbitrary
measurement efficiency and thermal occupation, but cannot
remove the fluctuations remaining within each conditional
state. 

The favorable scaling for the Hamiltonian in
Eq.~\eqref{eq:latticeHamiltonian} reflects the difficulty of engineering genuine global entanglement when restricted to this class of Hamiltonians. It nonetheless
illustrates the surprising computational simplicity that can emerge in
measurement-aligned, multimode non-Gaussian dynamics.

Unlike unmonitored amplitude damping, number dephasing mixes the measured
and conjugate quadratures, so the position marginal generally no longer
closes independently of the momentum hierarchy. Phase noise also affects
nonlinear squeezing directly \cite{Kala2022}. The loss-limited treatment
is most appropriate when phase fluctuations are negligible over the
relevant preparation and monitoring times. Improved oscillator coherence or independent suppression of
frequency noise could help reach this regime, although the
attainable benefit in practice remains to be established.

The method is targeted rather than universal.  Exact closure requires the
measured--conjugate separation identified above, and the coefficient count
still grows with mode number, polynomial degree, and requested momentum order.
It is therefore best suited to real-time tasks requiring a limited observable
sector, including nullifier tracking, commuting multimode feedforward, and
verification of non-Gaussian resources.  The same exact trajectories may also
serve as references for approximate filters and learned record processors.
\section{Conclusion}
\label{sec:conclusion}

We introduced conditional momentum moments as an exact finite representation
of selected observables in aligned-homodyne polynomial-phase dynamics.  Their
closure converts the measurement-induced upward moment hierarchy into a
finite system whose size is determined by the requested momentum order rather
than a Hilbert-space cutoff.  Applied to nonlinear nullifiers, the same
representation yields the exact causal displacement that removes all
record-accessible fluctuations of the conditional mean.

More broadly, these results show that real-time quantum filtering need not
always require real-time state reconstruction.  When the observable sector
has the structure identified here, propagating only the information required
by the eventual correction or diagnostic can retain the exact SME prediction
at substantially lower cost.  This provides a concrete route toward
low-latency monitoring, verification, and feedforward for polynomial-phase
and multimode non-Gaussian resources as experimental coherence and continuous
measurement capabilities improve.

\section*{Data Availability}

The source code, simulation parameters, and numerical data used to generate
all figures are publicly available in Ref.~\cite{EmersonCMMData}.
\appendix
\section{Normalization of the homodyne innovation}
\label{app:innovationCancellation}

From Eq.~\eqref{eq:wignerSMEClosure}, the stochastic Wigner update is
\begin{equation}
\left.\dd W_c\right|_{\rm stoch}
=
\sqrt{\frac{\eta\kappa}{2}}\,
\mathcal K_tW_c\,\dd W_t,
\qquad
\mathcal K_t=2(q-\bar q_t)+\partial_q.
\label{eq:appendixWignerInnovation}
\end{equation}
Because \(\mathcal K_t\) contains no \(p\), it commutes with the momentum
integrals defining \(\rho\) and \(J_n\).  Hence
\begin{align}
\left.\dd\rho\right|_{\rm stoch}
&=
\sqrt{\frac{\eta\kappa}{2}}\,
\mathcal K_t\rho\,\dd W_t,
\nonumber\\
\left.\dd J_n\right|_{\rm stoch}
&=
\sqrt{\frac{\eta\kappa}{2}}\,
\mathcal K_tJ_n\,\dd W_t.
\label{eq:appendixSectionInnovations}
\end{align}

At fixed \(q\), It\^o's quotient rule for \(m_n=J_n/\rho\) is
\begin{align}
\dd m_n
={}&
\frac{\dd J_n}{\rho}
-\frac{J_n}{\rho^2}\dd\rho
-\frac{\dd J_n\,\dd\rho}{\rho^2}
+\frac{J_n}{\rho^3}(\dd\rho)^2.
\label{eq:appendixItoQuotient}
\end{align}
The first two terms give the coefficient of \(\dd W_t\).  Using
\(J_n=m_n\rho\),
\begin{align}
\left.\dd m_n\right|_{\dd W}
={}&
\sqrt{\frac{\eta\kappa}{2}}\,
\frac{
\mathcal K_t(m_n\rho)-m_n\mathcal K_t\rho
}{\rho}\,\dd W_t.
\label{eq:appendixNoiseBeforeCancellation}
\end{align}
Expanding \(\mathcal K_t\) and applying the product rule,
\begin{align}
&\mathcal K_t(m_n\rho)-m_n\mathcal K_t\rho
\nonumber\\
={}&
2(q-\bar q_t)m_n\rho+\partial_q(m_n\rho)
\nonumber\\
&-
m_n\!\left[2(q-\bar q_t)\rho+\partial_q\rho\right]
\nonumber\\
={}&
\rho\,\partial_qm_n.
\label{eq:appendixExplicitNoiseCancellation}
\end{align}
The common multiplicative update and the derivative of the normalization
therefore cancel, leaving
\begin{equation}
\left.\dd m_n\right|_{\dd W}
=
\sqrt{\frac{\eta\kappa}{2}}\,
\partial_qm_n\,\dd W_t.
\label{eq:appendixNoiseCancellation}
\end{equation}

The last two terms in Eq.~\eqref{eq:appendixItoQuotient} give the associated
quadratic-variation drift:
\begin{align}
\left.\dd m_n\right|_{\rm qv}
={}&
-\frac{\eta\kappa}{2}
\frac{(\mathcal K_tJ_n)(\mathcal K_t\rho)}{\rho^2}\,\dd t
\nonumber\\
&+
\frac{\eta\kappa}{2}
\frac{J_n(\mathcal K_t\rho)^2}{\rho^3}\,\dd t.
\label{eq:appendixQuadraticVariation}
\end{align}
Equation~\eqref{eq:appendixExplicitNoiseCancellation} implies
\begin{equation}
\mathcal K_tJ_n
=
m_n\mathcal K_t\rho+\rho\,\partial_qm_n.
\label{eq:appendixOperatorDecomposition}
\end{equation}
Substituting this relation and \(J_n=m_n\rho\) into
Eq.~\eqref{eq:appendixQuadraticVariation} cancels the terms proportional to
\(m_n(\mathcal K_t\rho)^2\), giving
\begin{align}
\left.\dd m_n\right|_{\rm qv}
&=
-\frac{\eta\kappa}{2}
\frac{\mathcal K_t\rho}{\rho}\,
\partial_qm_n\,\dd t
\nonumber\\
&=
-\frac{\eta\kappa}{2}
\left[
2(q-\bar q_t)+\partial_q\ln\rho
\right]
\partial_qm_n\,\dd t.
\label{eq:appendixItoDrift}
\end{align}

The contribution produced by normalizing the common homodyne innovation is
therefore
\begin{align}
\left.\dd m_n\right|_{\rm hom}
={}&
\sqrt{\frac{\eta\kappa}{2}}\,
\partial_qm_n\,\dd W_t
\nonumber\\
&-
\frac{\eta\kappa}{2}
\left[
2(q-\bar q_t)+\partial_q\ln\rho
\right]
\partial_qm_n\,\dd t.
\label{eq:appendixNormalizedInnovation}
\end{align}
\section{Conversion between Weyl-ordered and fixed-order moments}
\label{app:orderingConversion}

The Wigner representation naturally returns Weyl-ordered observables.  For
nonnegative integers \(r,n\), define
\begin{equation}
\{Q^rP^n\}_{\rm W}
=\frac{1}{\binom{r+n}{r}}
\sum_{\pi\in\mathfrak W_{r,n}}\pi(Q^rP^n),
\label{eq:appendixWeylDefinition}
\end{equation}
where \(\mathfrak W_{r,n}\) is the set of distinct words containing \(r\)
copies of \(Q\) and \(n\) copies of \(P\).  Thus, for example,
\(\{QP\}_{\rm W}=(QP+PQ)/2\).  Equation
\eqref{eq:observableReconstruction} follows because multiplication by
\(q^rp^n\) in the Wigner integral represents this fully symmetrized operator.

The ordering conversion can be obtained from the
Baker--Campbell--Hausdorff identity
\begin{equation}
e^{i(sQ+tP)}
=e^{isQ}e^{itP}e^{ist/2}.
\label{eq:appendixBCHOrdering}
\end{equation}
Equating powers of \(s\) and \(t\) gives the triangular relations
\begin{align}
\{Q^rP^n\}_{\rm W}
={}&\sum_{k=0}^{\min(r,n)}
\binom{r}{k}\binom{n}{k}k!
\left(-\frac{i}{2}\right)^k
Q^{r-k}P^{n-k},
\label{eq:appendixWeylToFixed}\\
Q^rP^n
={}&\sum_{k=0}^{\min(r,n)}
\binom{r}{k}\binom{n}{k}k!
\left(\frac{i}{2}\right)^k
\{Q^{r-k}P^{n-k}\}_{\rm W}.
\label{eq:appendixFixedToWeyl}
\end{align}
For example, the Weyl-ordered moments returned by
Eq.~\eqref{eq:observableReconstruction} give
\begin{align}
QP
&=\{QP\}_{\rm W}+\frac{i}{2},
\nonumber\\
Q^2P
&=\{Q^2P\}_{\rm W}+i\{Q\}_{\rm W}.
\label{eq:appendixOrderingExamples}
\end{align}
Other fixed orderings, such as placing all \(P\)'s to the left, follow from
the same generating function or repeated use of \([Q,P]=i\).

Equations~\eqref{eq:appendixWeylToFixed} and
\eqref{eq:appendixFixedToWeyl} involve only moments with no greater momentum
or position order than the original monomial.  A single conversion contains
at most \(\min(r,n)+1\) terms and is performed only when observables are
reconstructed.  It therefore adds a finite triangular postprocessing step
and does not alter the coefficient count or the linear trajectory-length
scaling of the conditional-momentum hierarchy.  Weyl-ordered and fixed-order
mixed moments are consequently equivalent finite representations of the
observable sector used in this work.

\section{Gaussian bound for polynomial nullifier variance}
\label{app:gaussianNullifierBound}

We derive the optimized Gaussian benchmark
\(B_r(z)\) used in Eq.~\eqref{eq:nonlinearSqueezingRatio}.  Let a pure
single-mode Gaussian state have means \(q_0,p_0\) and covariance matrix
\begin{equation}
\bm V=
\begin{pmatrix}
v & s\\
s & (s^2+1/4)/v
\end{pmatrix},
\qquad v>0,
\label{eq:appendixPureGaussianCovariance}
\end{equation}
whose determinant is \(1/4\).  Since the nullifier variance is evaluated in
Weyl ordering, its moments equal averages over the positive Gaussian Wigner
function.  Because \(P+zQ^r\) is only linear in \(P\), the Weyl symbol of its
square is exactly \((p+zq^r)^2\); no additional star-product correction
appears.  We may therefore write
\begin{align}
Q&=q_0+X,
\nonumber\\
P&=p_0+\frac{s}{v}X+P_\perp,
\label{eq:appendixGaussianRegression}
\end{align}
where \(X\) and \(P_\perp\) are independent centered Gaussian variables with
variances \(v\) and \(1/(4v)\), respectively.

The displacement \(p_0\) does not affect the variance.  For fixed \(q_0\)
and \(v\), minimizing over the correlation \(s\) is equivalent to subtracting
the best linear predictor of \(Q^r\) from \(X\).  Hence
\begin{align}
\min_s\operatorname{Var}(P+zQ^r)
={}&\frac{1}{4v}
\nonumber\\
&+z^2\left[
\operatorname{Var}(Q^r)
-\frac{\operatorname{Cov}(Q^r,X)^2}{v}
\right].
\label{eq:appendixGaussianProjection}
\end{align}
Expanding $(q_0+X)^r$ in Gaussian Hermite polynomials shows that the
component orthogonal to the constant and linear terms is minimized at
$q_0=0$; for $r=2$, it is independent of $q_0$.  Thus $q_0=0$ may be
chosen without loss of optimality.

For a centered Gaussian variable,
\(\E[X^{2k}]=(2k-1)!!\,v^k\) and every odd moment vanishes.  If \(r\) is
even, \(\operatorname{Cov}(X^r,X)=0\), so
\begin{equation}
\operatorname{Var}(X^r)
=\left[(2r-1)!!-[(r-1)!!]^2\right]v^r.
\label{eq:appendixEvenGaussianResidual}
\end{equation}
If \(r\) is odd, \(\E[X^r]=0\) and
\begin{align}
\operatorname{Var}(X^r)
-\frac{\operatorname{Cov}(X^r,X)^2}{v}
={}&\Bigl[(2r-1)!!
\nonumber\\
&-r^2[(r-2)!!]^2\Bigr]v^r.
\label{eq:appendixOddGaussianResidual}
\end{align}
Both cases are summarized by the coefficient \(C_r\) in
Eq.~\eqref{eq:gaussianBoundCoefficient}.  Mixed Gaussian states cannot lower
the result: at fixed \(v\) and \(s\), the uncertainty relation replaces the
lower-right entry of Eq.~\eqref{eq:appendixPureGaussianCovariance} by a value
greater than or equal to \((s^2+1/4)/v\), adding nonnegative momentum noise.

It remains to minimize
\begin{equation}
\frac{1}{4v}+C_rz^2v^r.
\label{eq:appendixVarianceFunction}
\end{equation}
The stationary condition gives
\begin{equation}
v^{r+1}=\frac{1}{4rC_rz^2},
\label{eq:appendixOptimalVarianceCondition}
\end{equation}
and the positive solution is \(v_*(r,z)\) in
Eq.~\eqref{eq:optimalGaussianVariance}.  Substitution yields
\begin{equation}
B_r(z)=\frac{1}{4v_*}+C_rz^2v_*^r,
\label{eq:appendixFinalGaussianBound}
\end{equation}
which proves Eq.~\eqref{eq:gaussianNullifierBound}.

\section{Coefficient count for the cubic lattice}
\label{app:collectiveLattice}

For the lattice Hamiltonian in Eq.~\eqref{eq:latticeHamiltonian},
$[F,Q_i]=0$ for every mode. Thus, for product coherent initial
states, independent vacuum damping, and aligned local homodyne
monitoring, the position marginal evolves as in the single-mode
example of Sec.~\ref{sec:cubicExample}, with
\begin{align}
\langle Q_i\rangle_c
&=\bar q_{i,0}e^{-\kappa t/2},\\
\operatorname{Cov}_c(Q_i,Q_j)
&=\frac12\delta_{ij}.
\end{align}
The measured marginal is therefore a product Gaussian whose
evolution is known analytically.

The momentum dynamics, on the other hand, are governed by
the Hamiltonian force through the commutator
\begin{equation}
[F,P_i]=i\partial_{Q_i}F.
\end{equation}
For nearest-neighbor interactions on a square lattice, this
force depends only on the position quadratures of mode $i$
and its four nearest neighbors (fewer at the boundary).
Denote this set of modes by $S_i$. Since the Hamiltonian
does not alter the position dynamics and the remaining terms
in Eq.~\eqref{eq:multimodeConditionalHierarchy} preserve this
local dependence, $m_{\bm e_i}$ remains a quadratic polynomial
in at most five position coordinates.

To determine whether second momentum sections contain
additional information beyond the first, define
\begin{equation}
c_{ij}=m_{\bm e_i+\bm e_j}
-m_{\bm e_i}m_{\bm e_j},
\end{equation}
where $m_{\bm e_i}$ and $m_{\bm e_i+\bm e_j}$ are the
conditional momentum sections associated with $P_i$ and
$P_iP_j$, respectively, as defined in
Eq.~\eqref{eq:multimodeSections}.
Applying It\^o's product rule and subtracting the equations
for $m_{\bm e_i+\bm e_j}$ and $m_{\bm e_i}m_{\bm e_j}$ gives
\begin{align}
\dd c_{ij}
={}&-\kappa c_{ij}\,\dd t
+\frac{\kappa}{2}\delta_{ij}\,\dd t
\nonumber\\
&-\frac{\kappa}{2}
\sum_k u_k\partial_{u_k}c_{ij}\,\dd t
\nonumber\\
&+\frac{\kappa}{4}
\sum_k\partial_{u_k}^2c_{ij}\,\dd t
\nonumber\\
&+\frac{\kappa(1-\eta)}{2}
\sum_k
(\partial_{u_k}m_{\bm e_i})
(\partial_{u_k}m_{\bm e_j})\,\dd t
\nonumber\\
&+\sqrt{\frac{\eta\kappa}{2}}
\sum_k\partial_{u_k}c_{ij}\,\dd W_k,
\label{eq:latticeCenteredSection}
\end{align}
where $u_k=q_k-\bar q_{k,t}$.

If $S_i$ and $S_j$ do not overlap, the two first momentum
sections share no position coordinates, so the product of
derivatives vanishes. Since $c_{ij}(0)=0$ for distinct modes,
$c_{ij}$ remains zero and the second momentum section
factorizes as $m_{\bm e_i+\bm e_j}=m_{\bm e_i}m_{\bm e_j}$.
Additional coefficients are therefore needed only for pairs
with overlapping neighborhoods. For these pairs, the product
of derivatives is at most quadratic, so $c_{ij}$ remains a
quadratic polynomial in the coordinates of $S_i\cup S_j$.

A quadratic polynomial in $k$ variables contains
$\binom{k+2}{2}$ coefficients. The total number of retained
first- and second-order coefficients is therefore
\begin{align}
C_{\rm local}
={}&\sum_i\binom{|S_i|+2}{2}
\nonumber\\
&+\sum_{\substack{i\leq j\\
S_i\cap S_j\ne\varnothing}}
\binom{|S_i\cup S_j|+2}{2}.
\label{eq:localCoefficientCount}
\end{align}

Away from the boundary, the first momentum section and the
centered section with $i=j$ each involve five coordinates.
For distinct modes, neighborhoods overlap only when the
modes are separated by at most two lattice steps.
Nearest-neighbor and diagonal-neighbor pairs share two
coordinates, giving eight distinct coordinates, while pairs
separated by two sites along a lattice axis share one, giving
nine. There are four pairs of each type incident on an
interior mode, or two per mode after accounting for symmetry.
Thus,
\begin{align}
C_{\rm local}
&\sim M\left[
2\binom{7}{2}
+4\binom{10}{2}
+2\binom{11}{2}
\right]
\nonumber\\
&=(42+180+110)M
\nonumber\\
&=332M.
\end{align}
Accounting for the smaller neighborhoods at the boundaries
of an open $L\times L$ lattice with $L\geq4$ gives
\begin{equation}
C_{\rm local}=332M-830\sqrt{M}+430,
\end{equation}
where $M=L^2$. This is the \(O(M)\) coefficient count plotted in
Fig.~\ref{fig:latticeScaling}.
\section{Linear record-processing baselines}
\label{app:regressionBaselines}

The baselines in Fig.~\ref{fig:singleModeRecovery}(b) estimate the
endpoint conditional nullifier mean $\mu_T$ from the homodyne record.
The integrated-current baseline uses ordinary least squares with
the total integrated current $Y_T-Y_0$ as its input, while the
linear temporal baseline uses ridge regression with the integrated
current in 40 separate time bins as its inputs. The latter uses a
regularization penalty of $10^{-8}$, with each input standardized
using its training-set mean and standard deviation. Both models
include an intercept.

The first 10,000 simulated trajectories are used for training and
the remaining 10,000 for testing. For a predicted conditional
nullifier mean $\widehat\mu_T$, the corrected witness ratio is
\begin{equation}
\xi_{\rm lin}=
\frac{\E[V_{N,c}(T)]+\operatorname{Var}(\mu_T-\widehat\mu_T)}
{B_r(z)},
\end{equation}
where the expectation and variance are evaluated over the test
ensemble. The ignored-record and exact-filter results are evaluated
on the same ensemble. The reported values and 95\% confidence
intervals are obtained from the mean and the 2.5th and 97.5th
percentiles of 500 bootstrap resamples of complete test trajectories.

\bibliography{conditional_momentum_moments}

\end{document}